\documentclass[conference,twocolumn,final]{IEEEtran}
\usepackage{amsmath}
\usepackage{amsthm}
\usepackage{amssymb}
\usepackage{graphicx} % support the \includegraphics command and options
\usepackage{subfigure}
\usepackage{url}
\usepackage{booktabs} % for much better looking tables
\usepackage{array} % for better arrays (eg matrices) in maths
\usepackage{verbatim} % adds environment for commenting out blocks of text & for better verbatim
\usepackage{caption}
\usepackage{enumitem}

\newtheorem{theorem}{Theorem}%[section]

\newtheorem{definition}{Definition}
\newtheorem{lemma}{Lemma}
\newtheorem{problem}{Optimization Problem}

\newtheorem{corollary}{Corollary}

\usepackage{color}

\newcommand{\eq}[1]{Eq.~\eqref{#1}}
\usepackage{soul}
\setstcolor{red}

\usepackage{algorithm}
\usepackage{algorithmicx}
\usepackage{algpseudocode}

\begin{document}

\title{Femto-Caching with Soft Cache Hits: Improving Performance through Recommendation and Delivery of Related Content\\
%\rem{Soft Cache Hits: Improving Mobile Edge Caching Performance through Recommendation and Delivery of Alternative Content}
}
%\titlenote{Produces the permission block, and
%  copyright information}
%\subtitlenote{The full version of the author's guide is available as
%  \texttt{acmart.pdf} document}

\author{Pavlos Sermpezis\textsuperscript{1}, 
        Thrasyvoulos Spyropoulos\textsuperscript{2},  
        Luigi Vigneri\textsuperscript{2}, 
        and Theodoros Giannakas\textsuperscript{2}\\
        \textsuperscript{1}~ICS-FORTH, ~Greece, sermpezis@ics.forth.gr\\
        \textsuperscript{2}~EURECOM, ~France, first.last@eurecom.fr
        }

\maketitle

\begin{abstract}
Pushing popular content to cheap ``helper'' nodes (e.g., small cells) during off-peak hours has recently been proposed to cope with the increase in mobile data traffic. User requests can be served locally from these helper nodes, if the requested content is available in at least one of the nearby helpers. Nevertheless, the collective storage of a few nearby helper nodes does not usually suffice to achieve a high enough hit rate in practice. We propose to depart from the assumption of hard cache hits, common in existing works, and consider ``soft'' cache hits, where if the original content is not available, some related contents that are locally cached can be recommended instead. Given that Internet content consumption is entertainment-oriented, we argue that there exist scenarios where a user might accept an alternative content (e.g., better download rate for alternative content, low rate plans, etc.), thus avoiding to access expensive/congested links. We formulate the problem of optimal edge caching with soft cache hits in a relatively generic setup, propose efficient algorithms, and analyze the expected gains. We then show using synthetic and real datasets of related video contents that promising caching gains could be achieved in practice.
\end{abstract}

\section{Introduction}
\label{sec:intro}

%In the context of cellular networks, it is widely believed that aggressive densification, overlaying the standard macro-cell network with a large number of small cells (e.g., pico- or femto-cells), is a promising way of dealing with the ongoing data crunch. As this densification puts a tremendous pressure on the backhaul network, researchers have suggested storing popular content at the ``edge'', e.g., at small cells~\cite{femto}, user devices~\cite{femtoD2D,Hui-offloading,Whitbeck-offloading}, or vehicles acting as mobile relays~\cite{vigneri2016}. This could reduce congestion on capacity-limited backhaul links, as well as reduce the access latency to such content. 

Mobile edge caching has been identified as one of the five most disruptive enablers for 5G networks~\cite{Bocc2014}, both to reduce content access latency and to alleviate backhaul congestion. 
%caching in wireless presents additional opportunities but also challenges, often both related to the broadcast nature of the channel. As a result, a large number of research works has recently emerged along two main threads: (a) the cache placement problem: which contents to store on each edge cache to maximize hit rate, and (b) the cache-aided delivery problem: how caching policies can facilitate content delivery by creating coded transmission~\cite{Ali2014} or CoMP opportunities~\cite{PsouMobihoc2015}. 
However, the number of required storage points in future cellular networks will be orders of magnitude more than in traditional CDNs~\cite{borst2010} (e.g., 100s or 1000s of small cells corresponding to an area covered by a single CDN server today). As a result, the storage space per local edge cache must be significantly smaller to keep costs reasonable.
%, and would fit less than $0.1\%-0.01\%$ of the entire Internet catalog. 
Even if we considered a small subset of the entire Internet catalogue, e.g., a typical torrent catalog (1.5PB) or the Netflix catalogue (3PBs), with relatively skewed popularity distribution, and more than 1TB of local storage, cache hit ratios would still be low~\cite{Paschos-misconceptions, tutorial-sigmetrics}.

%This suggests that even though studies assuming a large (CDN-type) cache deep inside the core network give promising hit ratios~\cite{Erman2011}, only a tiny fraction of the constantly and exponentially increasing content catalog could realistically be stored at each edge cache. 

Additional caching gains have been sought by researchers, increasing the ``effective'' cache size visible to each user. This could be achieved by: (a) \emph{Coverage overlaps}, where each user is in range of multiple cells, thus having access to the aggregate storage capacity of these cells, as in the femto-caching framework~\cite{femto,poularakis2014toc}. (b) \emph{Coded caching}, where collocated users overhearing the same broadcast channel may benefit from cached content in other users' caches~\cite{Ali2014}. (c) \emph{Delayed content access}, where a user might wait up to a TTL for its request, during which time more than one cache (fixed~\cite{sermpezis2014%,Pavlos-Offload2016
} or mobile~\cite{Hui-offloading,Whitbeck-offloading,vigneri2016}) can be encountered. Each of these ideas could theoretically increase the cache hit ratio (sometimes significantly), but the actual practical gains might not suffice by themselves (e.g., due to high enough cell density required for (a), sub-packetization complexity in (b), and imposed delays in (c)). 

Existing edge caching approaches have a common goal: to deliver  \emph{every possible} content a user requests (if not from a local cache, then from a remote content server). While reasonable, this leads to many expensive cache misses that may potentially stifle the idea of edge caching. 
%, when the \emph{global} cache size becomes large enough. Nevertheless, taking the femto-caching case (a) as an example, a user will be in range of just a few (less than 10) cells in even very futuristic dense scenarios. If the effective cache size thus increased by $10\times$, cache hit ratios would increase, but would remain relatively low, as still less than $1\%-0.1\%$ of a single service catalog fits in the collective cache. Similar arguments could be made also for the other two methods.
%Similar conclusions can be drawn for the other methods (e.g., using the notation of~\cite{Ali2014} for coded caching, with $10$ concurrent user requests the key factor $KM/N$ in the above setup would be at most $10^{-3}/$, leading to a global caching gain of $\frac{1}{1+10^{-3}}$, a mere $0.1\%$ of extra gain).
%Operators, are thus left with a very costly dilemma: bear a huge cost for the backhaul infrastructure (e.g., fiber everywhere) or bear a huge cost for CDN-size storage at each and every small cell. 
We argue that, in an Internet which is becoming increasingly entertainment-oriented \emph{moving away from satisfying a given user request towards satisfying the user} could prove beneficial for caching systems. When a user requests a content not available in the local cache(s), a recommendation system could propose a set of \emph{related contents} that \emph{are} locally available. If the user accepts one of these contents, an expensive remote access could be avoided. We will use the term \emph{soft cache hit} to describe such scenarios.

%Let us consider two such example scenarios. (Example 1) A user requesting a content X, not available locally (e.g., a fan wanting to follow last \rem{week's ``NBA highlights''}), might be equally satisfied (in the best case) or not fully dissatisfied (in many cases), if he receives another content Y related to X (e.g.,  \rem{an ``NBA Top-10 plays/dunks/assists of the week'' video}). (Example 2) We could also consider users streaming content \emph{in sequence}, e.g., browsing YouTube videos back-to-back or listening to personalized radio. In this case the selected contents at each step are often \emph{recommended contents related to the previous one}. In many cases, the user might be almost equally happy with most of the recommended alternatives. A recommendation engine could prioritize related contents that are available in the cache, without a perceptible impact on user QoE. 

Although many users in today's cellular ecosystem might be reluctant to accept alternative contents, we believe there are a number of scenarios where soft cache hits could benefit both the user and the operator. As one example, a \emph{cache-aware} recommendation system could be a plugin to an existing application (e.g.,  the YouTube app), as shown in Fig.~\ref{fig:mobile-app-example}. The operator can give incentives to users to accept the alternative contents when there is congestion (e.g., \textit{zero-rating} services~\cite{t-mobile-music-freedom,neutral-net-neutrality}) or letting the user know that accessing content $X$ from the core infrastructure would be slow and choppy, while contents $A,B,C,...$ might have much better performance. The user can still reject the recommendation and demand the original content. In a second example (see Fig.~\ref{fig:mobile-app-example-second-model}) the operator might ``enforce'' an alternative (but related) content, e.g.: (i) making very low rate plans (currently offering little or no data) more interesting by allowing regular data access, except under congestion, at which time only locally cached content can be served; (ii) in developing areas~\cite{KioskNet} or when access to only a few Internet services is provided, e.g., the Facebook's Internet.org  (Free Basics) project~\cite{internet-org,free-basics}. %The user must then accept this content (but might be less happy with it than the original).

%\rem{It is not the goal of this paper to design such cache-aware recommendation apps or appropriate incentive schemes. However,} 
We believe such a system is quite timely, given the increased convergence of content providers with sophisticated recommendation engines (e.g., NetFlix and YouTube) and Mobile Network Operators (MNO), in the context of RAN Sharing~\cite{liang2015wireless,CellSlice2013}. Furthermore, the idea of soft cache hits is complementary and can be applied \emph{on top} of existing proposals for edge caching, like the ones described earlier. In a recent preliminary work~\cite{sch-chants-2016}, we have considered the idea of soft cache hits in a DTN context with mobile relays.
%\rem{Such operators would have direct access to both the recommendation system and (virtualized) storage space.} \rem{Alternatively, a soft cache hit system can be implemented as a \textit{partial} cooperation between a (traditional) MNO and a content provider (CP)\st{: the MNO informs the CP about the content cached in its infrastructure and the CP tunes its recommendation system, or the CP informs the MNO about the recommended contents and the MNO selects the caching policy}}.
Our goal in this paper is to develop the idea of soft cache hits in detail, applying it to standard mobile edge caching systems with cache cooperation (e.g.,~\cite{femto}). To our best knowledge, this is the first work to jointly consider related content recommendation/delivery gains and cache cooperation (e.g., femto-caching) gains. In this context, our main contributions are:
\begin{itemize}[leftmargin=*]
\item \emph{Soft Cache Hits (SCH) model:} We propose a generic model for mobile edge caching and alternative soft cache hits that can capture a number of interesting scenarios (both Fig.~\ref{fig:mobile-app-example} and Fig.~\ref{fig:mobile-app-example-second-model} - in Sections~\ref{sec:problem} and~\ref{sec:second-utility-model}, respectively).  
\item \emph{Single cache with SCH:} We formulate the problem of edge caching with SCH, in the context of a single cache. We show that the problem is NP-hard, and propose efficient approximation algorithms with provable performance (Section~\ref{sec:single}). 
\item \emph{Femto-caching with SCH:} We extend our framework to a femto-caching scenario with SCH, where each user might have access to more than one BS and local caches, as in~\cite{femto} (Section~\ref{sec:femto}). %While the femtocaching problem looks significantly more complex, we show that it also has a submodular objective subject to matroid constraints, and propose a caching algorithm with optimality guarantees.
\item \emph{Validation:} We show using both synthetic data and a real dataset of YouTube related videos that additional caching gains, e.g., \emph{on top of what femto-caching provides}, could be achieved in practice (Section~\ref{sec:sims}).
%% (Section~\ref{sec:sims}). 
\end{itemize}

Finally, we discuss related work and future research directions in Section~\ref{sec:related}, and conclude our paper in Section~\ref{sec:conclusions}.

\begin{figure}
\subfigure[]{\includegraphics[width=0.33\columnwidth]{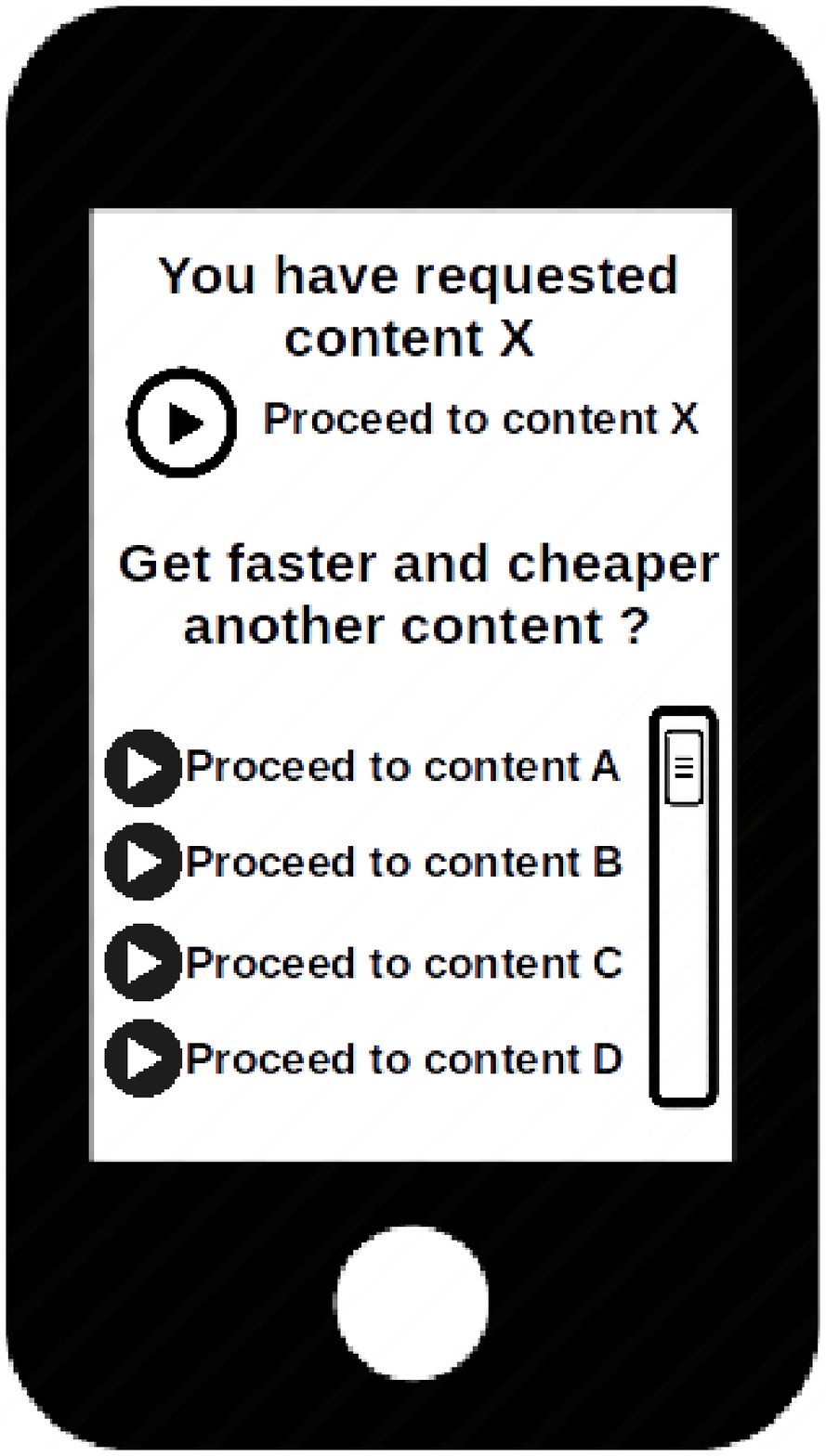}\label{fig:mobile-app-example}}
\hspace{0.1\linewidth}
\subfigure[]{\includegraphics[width=0.325\columnwidth]{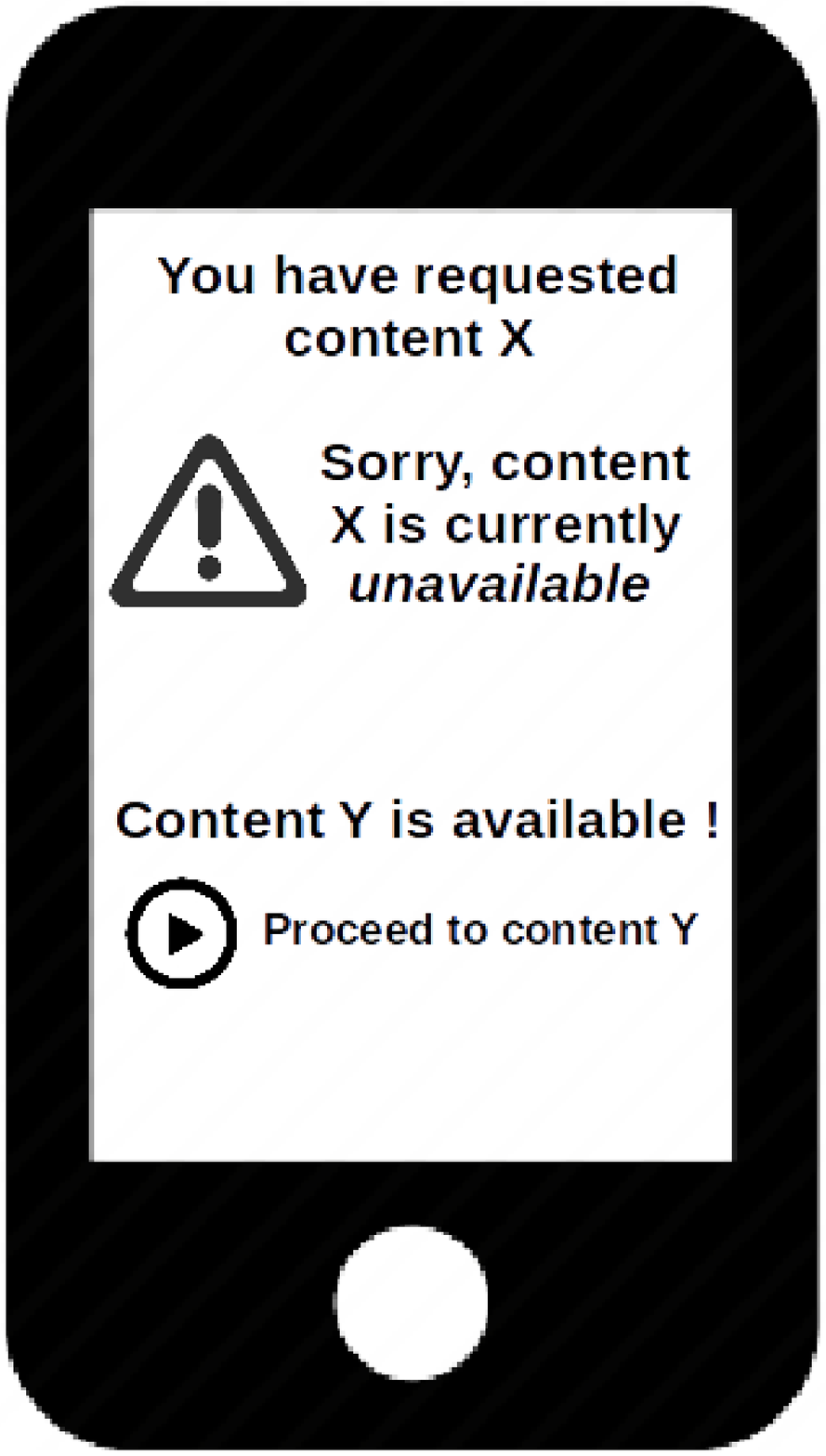}\label{fig:mobile-app-example-second-model}}
\caption{Mobile app example for \textit{Soft Cache Hits}: (a) related content \textit{recommendation} (that the user might not accept) , and (b) related content \textit{delivery}.}
\end{figure}

\section{Problem Setup}
\label{sec:problem}
\subsection{Network and Caching Model}
\textbf{Network Model}: Our network consists of a set of users $\mathcal{N}$ ($\| \mathcal{N} \| = N)$ and a set of small cells (SC) $\mathcal{M}$ ($\| \mathcal{M} \| = M)$. 
%We assume a dense network of SCs, so that each user $i \in \mathcal{N}$ might be able to communicate with more than one SCs $j \in \mathcal{M}$. 
Users are mobile and the SCs with which they associate might change over time. Since the caching decisions are taken in advance (e.g., the night before, as in~\cite{femto,poularakis2014toc}, or once per few hours or several minutes) it is hard to know the exact SC(s) each user will be associated at the time she requests a content%next day
. To capture user mobility, we propose a more generic model than the fixed bipartite graph of~\cite{femto}:
\begin{equation}
q_{ij} \doteq Prob\{\mbox{user i in range of SC j}\}, 
\end{equation} 
or, equivalently, the percentage of time a user $i$ spends in the coverage of SC $j$). Hence, deterministic $q_{ij}$ ($\in \{0,1\}$) captures the static setup of~\cite{femto}, while uniform $q_{ij}$ ($q_{ij} = q, \forall i,j$) represents the other extreme (no advance knowledge).

%Unlike the work in~\cite{femto}, we consider a more generic model of
%The connectivity relation between users and SCs is captured by a bipartite graph $\mathcal{G} = \{\mathcal{V},\mathcal{E}\}$, where an edge $\{i,j\} \in \mathcal{E}$ implies that user $i$ is ``covered'' by SC $j$ (i.e., can connect to and download content from $j$ with sufficient rate). We also denote with $\mathcal{G}(i) = \{j \in \mathcal{M}: \{i,j\} \in \mathcal{E} \}$ the set of all SCs that a user $i$ can access.  

\textbf{Content Model}: We assume each user requests a content from a catalogue $\mathcal{K}$ with $\| \mathcal{K} \| = K$ contents. A user $i\in\mathcal{N}$ requests content $k \in \mathcal{K}$ with probability $p_{k}^{i}$.\footnote{This also generalizes the standard femto-caching model~\cite{femto} which assumes same popularity per user. We can easily derive such a popularity \emph{per content} $p_{k}$ from $p_{k}^{i}$.} We will initially assume that all contents have the same size, and relax the assumption later.

\textbf{Cache Model (Baseline)}: We assume that each SC is equipped with storage capacity of $C$ contents (all our proofs hold also for different cache sizes). We use the integer variable  $x_{kj} \in \{0,1\}$ to denote if content $k$ is stored in SC $j$. In the traditional caching model, which we will consider as our baseline, if a user $i$ requests a content $k$ which is stored in some nearby SC, then the content can be accessed directly from the local cache and a \emph{cache hit} occurs. This type of access is considered ``cheap'', while a \emph{cache miss} leads to an ``expensive'' access (e.g., over the SC backhaul and core network).

%\add{For ease of reference, the main notation we use in our model and analysis is summarized in Table~\ref{table:notation}.}

{For ease of reference, the notation is summarized in Table~\ref{table:notation}.}

%\add{[say that we define the model as generic as possible, in order to be able to capture any scenario we have thought as well as scenarios we have not thought and might be proposed in the future. However, in some parts we consider sub-cases of our model that correspond better in today's network solutions and architectures.]}

\subsection{Soft Cache Hits} \label{sec:model-sch}

Up to this point the above model describes a baseline setup similar to the popular femto-caching framework~\cite{femto} (where we consider 0-1 cache hits, for simplicity, rather than access delay). The main departure in this paper is the following.

\textbf{Alternative Content Recommendation}: When the content a user asks for is not found in a local cache, a list of related contents out of the ones already cached is recommended to the user (see, e.g., Fig.~\ref{fig:mobile-app-example}). If a user selects one of them, a \textit{(soft) cache hit} occurs, otherwise there is a \textit{cache miss} and the network must fetch and deliver the original content.\footnote{Throughout our proofs, we assume, for simplicity, that the user can pick \emph{any} of the available cached contents; however, our analysis holds also when only a small subset of locally cached contents is recommended (e.g., the ones the recommender thinks are the most related for that user and for that original request).} Whether a user accepts an alternative content or not depends both on the content (how related it is) and on the user. This is captured in the following:

%\textbf{Content Utility}: \rem{Recommending an alternative content to a user does not mean that she accepts it as well. The acceptance of a recommended content might depend on (i) its similarity with the requested content, (ii) the type of the content (e.g., alternative contents for videos, songs, or news articles are easier to be found compared to alternatives for specific websites), (iii) the user interests/profile, (iv) the flexibility of the user, (v) etc. Hence, different alternative contents might have different probabilities to be selected by a user, and/or different users might be more or less probable to select a recommended content. To capture all these possible scenarios, we model in a generic way the relations between contents. Specifically, we define the (per user) \textit{content utility}:}
%\st{We assume that each content $k \in \mathcal{K}$ has a set of \emph{related contents}. Let $u_{kn}$ denote the utility a given user gets if she originally asks for content $k$ but instead receives content $n$, where $0 \le u_{kn} \le 1$ and $u_{kk} = 1, \forall i$.}

\begin{definition}\label{def:utility-as-probability}
A user $i$ that requests a content $k$ that is not available, accepts a recommended content $n$ with probability $u_{kn}^{i}$, where $0 \le u_{kn}^{i} \le 1$, and $u_{kk}^{i} = 1, \forall i,k$.
% (and rejects it with probability $1-u_{kn}^{i}$)
\end{definition}
The utilities define a content relation matrix $\mathbf{U^{i}} = \{u_{kn}^{i}\}$ for each user. Note that the above model captures the scenario of Fig.~\ref{fig:mobile-app-example}. We will use this model throughout Sections~\ref{sec:single} and ~\ref{sec:femto} to develop most of our theory. However, in Section~\ref{sec:second-utility-model}, we will modify our model to also analyze the scenario of Fig.~\ref{fig:mobile-app-example-second-model}, which we will refer to as \emph{Alternative Content Delivery}. 

Per user utilities $u_{kn}^{i}$ could be estimated from past statistics, and/or user profiles as usually done by standard recommendation algorithms~\cite{survey-collaborative-filtering}. In some cases, the system might have a coarser view of these utilities (e.g., item-item recommendation~\cite{amazon-recommendations}). We develop our theory and results the most generic case of Definition~\ref{def:utility-as-probability}, but we occasionally refer to the following two subcases:  
%However, in a real implementation, the system might have only a coarser view, e.g., estimations of the average (over all users) utility $u_{kn}$ of two contents (based on their similarities, past statistics, etc.). For instance, two examples of \textit{sub-cases} of Definition~\ref{def:utility-as-probability} are}
\begin{description}
\item[Sub-case 1:] The system does not know the exact utility $u_{kn}^{i}$ for each node $i$, but only how they are distributed among all nodes, i.e., the distributions $F_{kn}(x)\equiv P\{u_{kn}^{i}\leq x\}$.
%\begin{equation}
%F_{kn}(x)\equiv P\{u_{kn}^{i}\leq x\}
%\end{equation}
\item[Sub-case 2:] The system knows only the \textit{average utility} $u_{kn}$ per content pair $\{k,n\}$.%, where $u_{kn} \equiv \frac{1}{N}\cdot \sum_{i=1}^{N}u_{kn}^{i}$.
%\begin{equation}
%u_{kn} \equiv \frac{1}{N}\cdot \sum_{i=1}^{N}u_{kn}^{i}
%\end{equation}
\end{description}

%To facilitate the application of our results to systems with coarser knowledge of utilities, we also provide the needed modifications to our theory for the sub-cases 1 and 2 presented above.}

\begin{table}
\caption{{Important Notation}}\label{table:notation}
\begin{small}
\begin{tabular}{|l|l|}
\hline
{$\mathcal{N}$}		&{set of users ($\| \mathcal{N} \| = N)$)}\\
\hline
{$\mathcal{M}$}		&{set of small cells (SC)  ($\| \mathcal{M} \| = M)$)}\\
\hline
{$C$}				&{storage capacity of a SC}\\
\hline
{$q_{ij}$}			&{probability user $i$ in range of SC $j$}\\
{}					&{(or, $q_{i}$ for the single-cache case)}\\
\hline
{$\mathcal{K}$}		&{set of contents ($\| \mathcal{K} \| = K)$)}\\
\hline
{$p_{k}^{i}$}		&{probability user $i$ to request content $k$}\\
\hline
{$x_{kj}$}			&{content $k$ is stored in SC $j$ ($x_{kj}=1$) or not ($x_{kj}=0$)}\\
{}					&{(or, $x_{k}$ for the single-cache case)}\\
\hline
{$u_{kn}^{i}$}		&{utility of content $n$ for a user $i$ requesting content $k$}\\
\hline
{$F_{kn}(x)$}		&{distribution of utilities $u_{kn}^{i}$, $F_{kn}(x) = P\{u_{kn}^{i}\leq x\}$}\\
\hline
{$u_{kn}$}			&{average utility for content pair $\{k,n\}$ (over all users)}\\
\hline
{$s_{k}$}			&{size of content $k$}\\\hline
\end{tabular}
\end{small}
\end{table}

\section{Single Cache with Soft Cache Hits}
\label{sec:single}
%\add{[Pavlos]: (a) I formulated Problem 1, shown that it is NP-hard (Lemma 3.1), submodular and monotone (Lemma 3.2), proposed a greedy algorithm (Alg.1), and shown that Alg.1 is 1-1/e optimal (Th.3.3). \\(b) I formulated Problem 2, shown that it is NP-hard (Th 3.5), proposed two greedy algorithms (Alg.2 and 3), and shown the optimality of Alg 2 and 3 (Th.5.3).}

In order to better understand the impact of the related content matrices $\mathbf{U^{i}}$ on caching performance, we first consider a scenario where a user $i$ is served by a single base station, i.e., $\sum_j q_{ij} = 1, \forall i$ (i.e., each user is associated to exactly one BS, but we might still not know in advance which). Such a scenario is in fact relevant in today's networks, where the cellular network first chooses a single BS to associate a user to (e.g., based on signal strength), and then the user makes its request~\cite{LTE-Sesia}. In that case, we can optimize each cache independently. We can also drop the second index for both the storage variables $x_{kj}$ and connectivity variables $q_{ij}$, to simplify notation. 
%, \st{since every user has the same probabilistic demand $p_{k}$}\footnote{\st{This applies also if each BS sees a different popularity distribution $p_{k}$ or each user has a different distribution, in which case we can first see which users associate with which BS and then derive the total $p_{k}$ for a given BS.}}. 
%\rem{Since we refer to a single SC $j\in\mathcal{M}$, in this section -for notation simplicity- we drop the subscript $j$ from the quantities $x_{kj}$ and $q_{ij}$, and use the notation $x_{k}$ to indicate whether content $k$ is stored in the considered cache, and $q_{i}$ to indicate the probability a user $i$ to be served by the considered SC.} %\st{Storage variables (our control) are now denoted as $x_{k}$, indicating whether content $k$ is stored in the considered cache.} 
%\add{Finally, we consider a content relation graph $\mathbf{U}$ of ``Type 1'' (see Section~\ref{sec:model-sch}). We generalize our framework to ``Type 2'' content relations in Section~\ref{sec:extensions}}.

\subsection{Soft Cache Hit Ratio}\label{sec:schr-single}

A request (from a user to the SC) for a content $k\in\mathcal{K}$ would result in a (standard) cache hit only if the node stores content $k$ in the cache, i.e., if $x_{k}=1$. Hence, the (baseline) \textit{cache hit ratio} for this request is simply
\begin{equation*}
CHR(k) = x_{k}
\end{equation*}

If we further allow for soft cache hits, the user might be also satisfied by receiving a different content $n\in\mathcal{K}$. The probability of this event is, by Definition~\ref{def:utility-as-probability}, equal to $u_{kn}^{i}$. The following Lemma derives the total cache hit rate in that case.
\begin{lemma}[Soft Cache Hit Ratio (SCHR)]\label{thm:lemma-schr-single}
Let $SCHR$ denote the expected cache hit ratio for a single cache (including regular and soft cache hits), among all users. Then, 
\begin{equation}\label{eq:schr-single-defintion}
SCHR = \sum_{i=1}^{N}\sum_{k=1}^{K} p_{k}^{i}\cdot q_{i} \cdot \left(1 - \prod_{n=1}^{K} \left(1- u_{kn}^{i} \cdot x_{n}\right)\right).
\end{equation}
%\rem{where the quantity $p_{k}^{i}\cdot q_{i}$ denotes the requests for content $k$ by user $i$ to the considered SC.}
\end{lemma}
\begin{proof}
The probability of satisfying a request for content $k$ by user $i$ with related content $n$ is $P\{n|k,i\} = u_{kn}^{i}\cdot x_{n}$, 
since $u_{kn}^{i}$ gives the probability of acceptance (by definition), and $x_{n}$ denotes if content $n$ is stored in the cache (if the content is not stored, then $P\{n|k,i\}=0$). Hence, it follows easily that the probability of a \textit{cache miss}, when content $k$ is requested by user $i$, is given by $\prod_{n=1}^{K}(1- u_{kn}^{i} \cdot x_{n})$. The complementary probability, defined as the \emph{soft cache hit ratio} (SCHR), is then
\begin{equation}\label{eq:def-schr}
SCHR(i,k,\mathbf{U}) = 1-\prod_{n=1}^{K}(1- u_{kn}^{i} \cdot x_{n}).
\end{equation}
%In other words, the $SCHR(k,\mathbf{U})$ corresponds to the \emph{expected} cache hit ratio (the exact cache hit ratio cannot be known in advance, since the user decides a posteriori whether she accepts any of the recommended contents).
Summing up over all users that might be associated with that BS (with probability $q_{i}$) and all contents that might be requested ($p_{k}^{i}$) gives us Eq.(\ref{eq:schr-single-defintion}) 
\end{proof}

%which is $SCHR(i,\mathbf{U})=1$ if $\displaystyle\max_{j=\{n \in\mathcal{K}:u_{kn}=1\}}\{x_{n}\}=1$ and $SCHR(i,\mathbf{U})=0$ otherwise.
%\st{In other words, if $\mathbf{U}$ is of type 1, then $u_{kn}$ is an indicator variable and the above quantity gives the actual soft cache hit ratio. However, when $\mathbf{U}$ is of type 2, $SCHR(k,\mathbf{U})$ corresponds to the \emph{expected} soft cache hit ratio. In any case, the above definition captures both, and will be used in the remainder of the paper.}

Lemma~\ref{thm:lemma-schr-single} can be easily modified for the the sub-cases 1 and 2 of Def.~\ref{def:utility-as-probability} presented in Section~\ref{sec:model-sch}. We state the needed changes in Corollary~\ref{thm:corollary-single-cache-subcases}.
\begin{corollary}\label{thm:corollary-single-cache-subcases}
Lemma~\ref{thm:lemma-schr-single} holds for the the sub-cases 1 and 2 of Def.~\ref{def:utility-as-probability}, by substituting in the expression of \eq{eq:schr-single-defintion} the term $u_{kn}^{i}$ with 

\vspace{-\baselineskip}
{\small
\begin{align}
u_{kn}^{i} &\rightarrow E[u_{kn}^{i}]\equiv\int\left(1-F_{kn}(x)\right)dx &~\text{(for sub-case 1)}\\
u_{kn}^{i} &\rightarrow u_{kn} &~\text{(for sub-case 2)}
\end{align}
}
\end{corollary}
\begin{proof}
The proof is given in Appendix~\ref{app:corollary-single-subcases}.
\end{proof}

%\rem{\underline{Remark}: For the sub-cases 1 and 2 of Def.~\ref{def:utility-as-probability} presented in Section~\ref{sec:model-sch}, it can be shown  that in the expression of \eq{eq:schr-single-defintion} for the SCHR we just need to substitute the term $u_{kn}^{i}$ with the term $E[u_{kn}^{i}]\equiv\int\left(1-F_{kn}(x)\right)dx$ (for sub-case 1) or $u_{kn}$  (for sub-case 2).}

\subsection{Optimal SCH for Same Content Sizes}\label{sec:single-equal-size}

The (soft) cache hit ratio depends on the contents that are stored in a SC. The network operator can choose the storage variables $x_{k}$ to maximize SCHR by solving the following optimization problem. 

%\st{Considering now the request probabilities $p_{k}$ and the capacity constraint gives us the following optimization problem, when soft cache hits are allowed.}

%A request from a user $n\in \mathcal{N}$ to a node (e.g., SC) $m\in\mathcal{M}$ for a content $i\in\mathcal{K}$ would result to a cache hit if the node stores the content in the cache, i.e., if $x_{im}=1$. Hence, the \textit{cache hit ratio} for this request is
%\begin{equation*}
%CHR(n,i,m) = x_{im}
%\end{equation*}
%
%When considering also soft cache hits, the user will be satisfied by receiving any content $n \in\mathcal{K}$ with $u_{kn}=1$ (by definition $u_{i,i}=1$), and the \textit{soft cache hit ratio} for this request can be expressed as
%\begin{equation*}
%SCHR(n,i,m,\mathbf{U}) = 1-\prod_{j=1}^{K}(1-x_{jm})^{u_{kn}}
%\end{equation*}
%which is $SCHR(n,i,m,\mathbf{U})=1$ if $\displaystyle\max_{j=\{n \in\mathcal{K}:u_{kn}=1\}}\{x_{jm}\}=1$ and $SCHR(n,i,m,\mathbf{U})=0$ otherwise.

\begin{problem} \label{problem:single-cache}
The optimal cache placement problem for a single cache with soft cache hits and content relations described by the matrix $\mathbf{U} = \{{u_{kn}^{i}}\}$, is

\vspace{-\baselineskip}
{\small
\begin{eqnarray}
\underset{X = \{x_{1}, ..., x_{K}\}}{\mbox{maximize }} \;  f(X) = \sum_{i=1}^{N}\sum_{k=1}^{K} p_{k}^{i}\cdot q_{i} \cdot \left(1 - \prod_{n=1}^{K} \left(1- u_{kn}^{i} \cdot x_{n}\right)\right), 
\label{eq:objective-single-cache}\\
 \sum_{k =1}^{K} x_{k} \le C.
\label{eq:constraint-single-cache}
\end{eqnarray}
}
\end{problem}

In the following, we prove that the above optimization problem is NP-hard (Lemma~\ref{lemma:single-hardness}), and study the properties of the objective function \eq{eq:objective-single-cache} (Lemma~\ref{lemma:single-submodular}) that allow us to design an efficient approximate algorithm (Algorithm~\ref{alg:greedy-single-cache}) with provable performance guarantees (Theorem~\ref{thm:greedy}).

\begin{lemma}\label{lemma:single-hardness} The Optimization Problem~\ref{problem:single-cache} is NP-hard.
\end{lemma}

\begin{lemma}\label{lemma:single-submodular}
The objective function of Eq.(\ref{eq:objective-single-cache}) is submodular and monotone.
\end{lemma}

The proofs for the previous two Lemmas can be found in Appendices~\ref{app:single-hardness} and ~\ref{app:single-submodular}, respectively.

We propose Algorithm~\ref{alg:greedy-single-cache} as a greedy algorithm for Optimization Problem~\ref{problem:single-cache}: to select the contents to be stored in the cache, we start from an empty cache (line~1), and start filling it (one by one) with the content that increases the most the value of the objective function (line~4), till the cache is full. The computation complexity of the algorithm is $O\left(C\cdot K\right)$, 
%where $C$ and $K$ are the cache and catalog sizes, respectively.
since the loop (lines~2-6) denotes $C$ repetitions, and in each repetition the objective function is evaluated $y$ times, where $K\geq y\geq K-C+1$.

The following theorem gives the performance bound for Algorithm~\ref{alg:greedy-single-cache}. 
\begin{theorem}\label{thm:greedy}
Let $OPT$ be the optimal solution of the Optimization Problem~\ref{problem:single-cache}, and $S^{*}$ the output of Algorithm~\ref{alg:greedy-single-cache}. Then, it holds that
%, i.e., $OPT = \max_{S:|S|\leq C} f(S)$
\begin{equation}
f(S^{*}) \geq \left(1-\frac{1}{e}\right) \cdot OPT
\end{equation}
%Consider the following greedy algorithm. Start with $S_{0}$ empty. Then, repeat 
%\begin{eqnarray}
%j & = &  \underset{i \in K \setminus S_{i-1}}{argmax}f(S_{i-1} \cup \{j\}) \\
%S_{i} & = &  S_{i-1} \cup \{j\},
%\end{eqnarray}
%until $S_{C}$ is obtained.
%This algorithm achieves a $\left(1-\frac{1}{e}\right)$-approximation in the worst case.
\end{theorem} 
\begin{proof}
Lemma~\ref{lemma:single-submodular} shows that the Optimization Problem~\ref{problem:single-cache} belongs to the generic category of maximization of submodular and monotone functions (\eq{eq:objective-single-cache}) with a cardinality constraint (\eq{eq:constraint-single-cache}). For such problems, it is known that the greedy algorithm achieves (in the worst case) a $\left(1-\frac{1}{e}\right)$-approximation solution~\cite{nemhauser-greedy-approximation,krause2012submodular}.
\end{proof}

While the above is a strict worst case bound, it is known that greedy algorithms perform quite close to the optimal in most scenarios. In Section~\ref{sec:sims} we show that this simple greedy algorithm can already provide interesting performance gains.

%\add{This is confirmed in our simulation section. Note: perhaps we could get OPT for some small problems. Or perhaps we can solve the continuous version of the problem, i.e. $x_{k}$ are real, and show that our greedy algorithm performs close to this continuous case, which is always at least as good as OPT. Not sure it will work though..} 

%\begin{theorem}
%What about a continuous relaxation with (pipage rounding) or other more recent approximation results? we can write them here. E.g., there is a $(1-\epsilon)$ approximation for polynomial algorithms (based on dynamic programming). these are the FPTAS type of algorithms. could we describe the respective one for our problem?
%\end{theorem}

\begin{algorithm}[t]%[h]
\caption{Greedy Algorithm for Optimization Problem~\ref{problem:single-cache}.\\\textit{computation complexity: $O\left(C\cdot K\right)$}}
\begin{algorithmic}[1]
\Statex \textbf{Input:} \textit{utility $\{u_{kn}^{i}\}$}, \textit{content demand $\{p_{k}^{i}\}$}, \textit{mobility $\{q_{i}\}$}, $\forall k,n\in\mathcal{K},~i\in\mathcal{N}$
\State $S_{0}\leftarrow\emptyset;~t\leftarrow 0$
\While {$t < C$}
	\State $t\leftarrow t+1$
	\State $n \leftarrow  \underset{\ell \in K \setminus S_{t-1}}{argmax}f(S_{t-1} \cup \{\ell\}) $
	\State $S_{t} \leftarrow  S_{t-1} \cup \{n\},$
\EndWhile
\State $S^{*}\leftarrow S_{t}$
\State\Return $S^{*}$
\end{algorithmic}
\label{alg:greedy-single-cache}
\end{algorithm}

\subsection{Optimal SCH for Different Content Sizes}\label{sec:single-different-size}
%\add{[Pavlos] (To motivate somehow why we attack separately the two variations of the problem:  )} 

Till now we have assumed that all contents have equal size. In practice, each content has a different size $s_{k}$ and the capacity $C$ of each cache must be expressed in Bytes. 
%When the variation in the content sizes is large, solving the optimization problem becomes more cumbersome; a greedy algorithm can perform arbitrarily bad, and algorithms that provide solutions as good as in the equal-size case are more computationally demanding. 
Additionally, if a user requests a video of duration $X$ and she should be recommended an alternative one of similar duration $Y$ (note that similar duration does not always mean similar size). While the latter could still be taken care of by the recommendation system (our study of a real dataset in Section~\ref{sec:sims} suggests that contents of different sizes might still be tagged as related), we need to revisit the optimal allocation problem:  the capacity constraint of Eq.(\ref{eq:constraint-single-cache}) is no longer valid, and Algorithm~\ref{alg:greedy-single-cache} can perform arbitrarily bad.

%We now generalize somewhat(???) our setup and assume that each content $k$ has different size $s_{k}$ and the capacity of a cache is $C$ Bytes (not contents, as before). The problem then becomes a set cover problem with a knapsack type constraint:

\begin{problem}\label{problem:single-knapsack}
The optimal cache placement problem for a single cache with soft cache hits and variable content sizes, and content relations described by the matrix $\mathbf{U} = \{u_{kn}^{i}\}$ is

\vspace{-\baselineskip}
{\small
\begin{eqnarray}
\underset{X = \{x_{1}, ..., x_{K}\}}{\mbox{maximize }} \; f(X) = \sum_{i=1}^{N}\sum_{k=1}^{K} p_{k}^{i}\cdot q_{i} \cdot \left(1 - \prod_{j=1}^{K} \left(1- u_{kn}^{i} \cdot x_{n}\right)\right), \\
\sum_{k =1}^{K} s_{k} x_{k} \leq C.
\end{eqnarray}
}
\end{problem}

\emph{Remark:} Note that the objective is still in terms of cache hit ratio, and does not depend on content size. This could be relevant when the operator is doing edge caching to reduce \emph{access latency} to contents (latency is becoming a core requirement in 5G), in which case a cache miss will lead to long access latency (to fetch the content from deep inside the network), for both small and large contents. %\st{If however an operator's main concern is the backhaul load, a modified objective should perhaps be used, where the above SCHR is weighed by the size of the respective content. We defer investigating this case to future work.}

The problem is a set cover problem variant with a knapsack type constraint. We propose approximation Algorithm~\ref{alg:fast-single-knapsack} for this problem, which is a ``fast greedy'' algorithm (based on a modified version of the greedy Algorithm~\ref{alg:greedy-single-cache}) and has complexity $O\left(K^{2}\right)$. 
% and can guarantee an $\frac{1}{2}\left(1-\frac{1}{e}\right)$-approximation of the optimal solution. 

%add{Algorithm~\ref{alg:complex-single-knapsack} can guarantee a better performance, $\left(1-\frac{1}{e}\right)$-approximation of the optimal solution. Algorithm~\ref{alg:complex-single-knapsack} first evaluates all the cases where the cache stores only 1 or 2 contents (only those that fit in the cache) and calculates which of these contents or tuples of contents, $S^{(1)}$, maximizes the value objective function (lines~1-2). Then, for each triple of contents that fit in the cache (line~3), it calculates the solution of the modified greedy algorithm starting from a cache containing this triple (rather than from an empty cache), and keeps the solution (corresponding to a triple) $S^{(2)}$ with the highest value of the objective function (lines~4-12). The returned solution, is the one between $S^{(1)}$ and $S^{(2)}$ that achieves a higher value of the objective function (lines~13-17). However, the improve performance guarantees come with a significant increase in the required computations, $O\left(K^{5}\right)$, which might not be feasible in a practical scenario when the catalog size $K$ is large.}

%\add{The following theorem formalizes the above discussion on the complexity of Optimization Problem~\ref{problem:single-knapsack} and the performance of Algorithms~\ref{alg:fast-single-knapsack} and~\ref{alg:complex-single-knapsack}.}

\begin{theorem}\label{lemma:single-knapsack}~\\
\noindent(1) The Optimization Problem~\ref{problem:single-knapsack} is NP-hard. \\
\noindent(2) Let $OPT$ be the optimal solution of the Optimization Problem~\ref{problem:single-knapsack}, and $S^{*}$ the output of Algorithm~\ref{alg:fast-single-knapsack}. Then, it holds that
%, i.e., $OPT = \max_{S:\sum_{i\in S} s_{i}\leq C} f(S)$,
\begin{equation}
f(S^{*}) \geq \frac{1}{2}\left(1-\frac{1}{e}\right) \cdot OPT
%&f(S^{**}) \geq \left(1-\frac{1}{e}\right) \cdot OPT
\end{equation}
\end{theorem} 
\begin{proof}
The proof can be found in Appendix~\ref{app:single-knapsack}.
\end{proof}

In fact, a polynomial algorithm with better performance $\left(1-\frac{1}{e}\right)$-approximation could be described, based on~\cite{max-submodular-knapsack}. However, the improved performance guarantees come with a significant increase in the required computations, $O\left(K^{5}\right)$, which might not be feasible in a practical scenario when the catalog size $K$ is large. We therefore just state its existence, and do not consider the algorithm further in this paper (the algorithm can be found in Appendix~\ref{app:algorithm-single-1/e}).

\begin{algorithm}[t]%[h]
\caption{$\frac{1}{2}\cdot \left(1-\frac{1}{e}\right)$-approximation Algorithm for Optimization Problem~\ref{problem:single-knapsack}.\\{\textit{computation complexity: $O\left(K^{2}\right)$}}}
\begin{algorithmic}[1]
\Statex {\textbf{Input:} \textit{utility $\{u_{kn}^{i}\}$}, \textit{content demand $\{p_{k}^{i}\}$}, \textit{content size $\{s_{k}\}$}, \textit{mobility $\{q_{i}\}$}, $\forall k,n\in\mathcal{K},~i\in\mathcal{N}$}
\State $S^{(1)}\leftarrow$\textproc{ModifiedGreedy}($\emptyset$,[$s_{1}$, $s_{2}$,...,$s_{k}$])
\State $S^{(2)}\leftarrow$\textproc{ModifiedGreedy}($\emptyset$,[$1$, $1$,...,$1$])
\If{$f(S^{(1)})>f(S^{(2)})$}
	\State $S^{*}\leftarrow  S^{(1)}$
\Else
	\State $S^{*}\leftarrow  S^{(2)}$
\EndIf
\State\Return $S^{*}$
\Statex {}
\Function{ModifiedGreedy}{$S_{0}$,[$w_{1}$, $w_{2}$,...,$w_{k}$]}
\State $\mathcal{K}^{(1)}\leftarrow \mathcal{K};~c\leftarrow 0;~t\leftarrow 0$
\While {$\mathcal{K}^{(1)} \neq \emptyset$}
	\State $t\leftarrow t+1$
	\State $n \leftarrow  \underset{\ell \in K \setminus S_{t-1}}{argmax} \frac{f(S_{t-1} \cup \{\ell\})}{w_{\ell}} $
	\If {$c+w_{n}\leq C$}
		\State $S_{t}\leftarrow  S_{t-1} \cup \{n\}$
		\State $c\leftarrow c+w_{n}$
	\Else 
		\State $S_{t}\leftarrow  S_{t-1}$
	\EndIf
	\State $\mathcal{K}^{(1)} \leftarrow \mathcal{K}^{(1)}\backslash\{n\}$ 
\EndWhile
\State \Return $\leftarrow S_{t}$
\EndFunction
\end{algorithmic}
\label{alg:fast-single-knapsack}
\end{algorithm}

\begin{comment}

\begin{algorithm}%[h]

\caption{$\left(1-\frac{1}{e}\right)$-approximation Algorithm for Optimization Problem~\ref{problem:single-knapsack}.}
\begin{algorithmic}[1]
\State $A\leftarrow \left\{S\subseteq \mathcal{K}:~|S|<3~\text{and}~\sum_{i\in S}s_{i}\leq C\right\}$
\State $S^{(1)}\leftarrow  \underset{S \in A}{argmax} f(S)$
\Statex
\State $B\leftarrow \left\{S\subseteq \mathcal{K}:~|S|=3~\text{and}~\sum_{i\in S}s_{i}\leq C\right\}$
\State $S^{(2)}\leftarrow \emptyset$
\For {$S\in B$}
	\State $U\leftarrow \mathcal{K}\backslash S$
	\State $W\leftarrow$ list$(w_{i}),~\forall i\in U$
	\State $H\leftarrow$\textproc{ModifiedGreedy}($U$,[W])
	\If {$f(H)>f(S^{(2)})$}
		\State $S^{(2)}\leftarrow H$
	\EndIf 
\EndFor
\Statex
\If{$f(S^{(1)})>f(S^{(2)})$}
	\State $S^{*}\leftarrow  S^{(1)}$
\Else
	\State $S^{*}\leftarrow  S^{(2)}$
\EndIf
\State\Return $S^{*}$
\end{algorithmic}
\label{alg:complex-single-knapsack}
\end{algorithm}

\end{comment}

%\subsection{Performance Improvement}

%\add{I guess there is no time/space to say sth analytical about the performance gains here. But maybe we could just provide some numerical results here (the subset of results I mentioned in the mail) for various popularities, cache sizes and U graphs. As these are not really simulation results, it wouldn't make sense to put them in the simulation section. There, we can focus on the femtocaching case.}

%\begin{itemize}
%\item \emph{Performance Improvement}: Given a caching

\section{Femtocaching with Related Content Recommendation}
\label{sec:femto}
%\add{[Pavlos]}

%\add{There exists a better $\left(1-\frac{1}{e}\right)$-approximation following the ``multilinear extension'' approach~\cite{calinescu-max-submodular-matroid} (or~\cite[Sec. 3.2]{krause2012submodular}. However, I am not sure about my understanding of the method; I did not manage to find out how to express the respective algorithm. In particular the multilinear expression method consists of two steps: (a) solve the continuous equivalent problem with a "greedy" method~\cite{calinescu-max-submodular-matroid,krause2012submodular}, and then (b) do ``pipage rounding''. I don't know how to express the part of the algorithm for (a); the part (b) is trivial in our case (if I am not mistaken), it is just a simple rounding of the continuous solution, since our matroid is a partition matroid~\cite{calinescu-max-submodular-matroid}.}

%\add{In fact, there is also another method that can give $\left(1-\frac{1}{e}\right)$-approximation~\cite{filmus2014monotone}, however, I did not understand it and could not express an algorithm based on it.}

Building on the results and analytical methodology of the previous section for the optimization of a single cache with soft cache hits, we now extend our setup to consider the complete problem with cache overlaps (referred to as ``femtocaching''~\cite{femto}). Note however that we \textit{do} consider user mobility, through variables $q_{ij}$, unlike some works in this framework that often assume static users. Due to space limitations, we focus on the case of fixed content sizes.

In this scenario, a user $i\in\mathcal{N}$ might be covered by more than one small cells $j \in \mathcal{M}$, i.e. $\sum_{j} q_{ij} \ge 1, \forall i$. A user is satisfied, if she receives the requested content $k$ or \emph{any} other related content (that she will accept), from \emph{any} of the SCs within range. Hence, similarly to \eq{eq:def-schr}, the total cache hit rate SCHR (that includes regular and soft cache hits) can be written as
\begin{equation}
SCHR(i,k,\mathbf{U}) =  1 - \prod_{j=1}^{M} \prod_{n=1}^{K} \left(1- u_{kn}^{i}\cdot x_{nj}\cdot q_{ij}\right)
\end{equation}
since for a cache hit a user $i$ needs to be in the range of a SC $j$ (term $q_{ij}$) that stores the content $n$ (term $x_{nj}$), and accept the recommended content (term $u_{kn}^{i}$).

Considering (i) the request probabilities {$p_{k}^{i}$}, (ii) every user in the system, and (iii) the capacity constraint, gives us the following optimization problem.

\begin{problem}\label{problem:femto-related}
The optimal cache placement problem for the femtocaching scenario with soft cache hits and content relations described by $\mathbf{U} = \{u_{kn}\}$ is

\vspace{-\baselineskip}
{\small
\begin{align}
\underset{X = \{x_{11}, ..., x_{KM}\}}{\mbox{maximize }}  f(X) = \sum_{i=1}^{N} \sum_{k=1}^{K}   &   p_{k}^{i} \left(1 - \prod_{j=1}^{M} \prod_{n=1}^{K} \left(1- u_{kn}^{i}\cdot x_{nj}\cdot q_{ij}\right)\right), \label{eq:objective-femto-related}\\
&  \sum_{k =1}^{K} x_{kj} \leq C, ~~~\forall j \in \mathcal{M}.\label{eq:constraint-femto-related}
\end{align}
}
\end{problem}

The following lemma states the complexity of the above optimization problem, as well as its characteristics that allow us to design an efficient approximation algorithm. 
\begin{lemma}\label{lemma:femto-np-matroid}~\\
(1) The Optimization Problem~\ref{problem:femto-related} is NP-hard, \\
(2) with submodular and monotone objective function (\eq{eq:objective-femto-related}) and a matroid constraint (\eq{eq:constraint-femto-related}).
\end{lemma}
\begin{proof}
%We prove Lemma~\ref{lemma:femto-np-matroid} by extending the basic ideas of the single-cache case, and following a similar methodology as in the proofs of Lemmas~\ref{lemma:single-hardness} and~\ref{lemma:single-submodular}.
The proof is given in Appendix~\ref{app:femto-NP-matroid}.
\end{proof}

Lemma~\ref{lemma:femto-np-matroid} states that the Optimization Problem~\ref{problem:femto-related} is a maximization problem with a submodular function and a matroid constraint. For this type of problems, a greedy algorithm can guarantee an $\frac{1}{2}$-approximation of the optimal solution~\cite{krause2012submodular}. We propose such a greedy algorithm in Algorithm~\ref{alg:greedy-femto}, which is of computational complexity $O\left(K^{2} M^{2}\right)$.

\begin{theorem}\label{th:femto-related}
Let $OPT$ be the optimal solution of the Optimization Problem~\ref{problem:femto-related}, and $S^{*}$ the output of Algorithm~\ref{alg:greedy-femto}. Then, it holds that
\begin{equation}
f(S^{*}) \geq \frac{1}{2}\cdot OPT\\
\end{equation}
\end{theorem} 

%In Algorithm~\ref{alg:greedy-femto} we begin with an empty placement $S_{0}$ (line~1), which means that the number of contents for each cache $c_{j}$ is initially zero (lines~2-4). We start placing contents in the caches, by selecting (one at a time) the tuple \{content,cache\} that increases the most the value of the objective function (line~7). 

%\add{Again, we can make additional extensions:
%\begin{itemize}
%\item knapsack constraints, i.e. contents have different size. need to check if this is approximable as before
%\item continuous relaxation + pipage rounding (similar to the femtocaching paper). 
%\item there is also the continuous multilinear relaxation by Vondrak (the guy whose slides I sent you)...not sure if it's the same as the above.
%\item a possible interesting problem could be to find the minimum cache size required to achieve a given cache hit ratio. because submodular minimization problems tend to be solved exactly with a continuous relaxations (through the Lovasz extension), maybe we could say sth cute here, which is also out of the standard "femtocaching" framework (e.g. this is a cache sizing problem, rather than placement).
%\end{itemize}
%}

\begin{algorithm}[t]%[h]
\caption{Greedy Algorithm for Optimization Problem~\ref{problem:femto-related}.\\{\textit{computation complexity: $O\left(K^{2}\cdot M^{2}\right)$}}}
\begin{algorithmic}[1]
\Statex {\textbf{Input:} \textit{utility $\{u_{kn}^{i}\}$}, \textit{content demand $\{p_{k}^{i}\}$}, \textit{mobility $\{q_{ij}\}$}, $\forall k,n\in\mathcal{K},~i\in\mathcal{N},~j\in\mathcal{M}$}
\State $A\leftarrow \mathcal{K}\times \mathcal{M};~S_{0}\leftarrow\emptyset;~t\leftarrow 0$
\For {$j\in \mathcal{M}$}
	\State $c_{j}\leftarrow 0$
\EndFor
%\State $S_{0}\leftarrow\emptyset;~t\leftarrow 0$
\While {$A \neq \emptyset$}
	\State $t\leftarrow t+1$
	\State $(n,j) \leftarrow  \underset{(k,\ell) \in A}{argmax} ~~f(S_{t-1} \cup \{(k,\ell)\}) $
	\Statex \Comment where, n: content; j: cache/SC
	\If {$c_{j}+1\leq C$}
		\State $c_{j}\leftarrow c_{j}+1$
		\State $S_{t} \leftarrow  S_{t-1} \cup \{(n,j)\}$
	\Else
		\State $S_{t} \leftarrow  S_{t-1} \cup \{(n,j)\},$
	\EndIf
	\State $A\leftarrow A \backslash \{(n,j)\}$
\EndWhile
\State $S^{*}\leftarrow S_{t}$
\State\Return $S^{*}$
\end{algorithmic}
\label{alg:greedy-femto}
\end{algorithm}

Submodular optimization problems have received considerable attention recently, and a number of sophisticated approximation algorithms have been considered (see, e.g.,~\cite{krause2012submodular} for a survey). For example, a better $\left(1-\frac{1}{e}\right)$-approximation can be found following the ``multilinear extension'' approach~\cite{calinescu-max-submodular-matroid}, based on a continuous relaxation and pipage rounding. A similar approach has also been followed in the original femto-caching paper~\cite{femto}. Other methods also exist that can give an $\left(1-\frac{1}{e}\right)$-approximation~\cite{filmus2014monotone}. Nevertheless, minimizing algorithmic complexity or optimal approximation algorithms are beyond the scope of this paper. Our goal instead is to derive {fast and efficient} algorithms (like greedy) that can handle the large content catalogues and content related graphs $\mathbf{U}$, and compare the performance improvement offered by soft cache hits. The worst-case performance guarantees offered by these algorithms are added value. 
  
%However, I am not sure about my understanding of the method; I did not manage to find out how to express the respective algorithm. In particular the multilinear expression method consists of two steps: (a) solve the continuous equivalent problem with a "greedy" method~\cite{calinescu-max-submodular-matroid,krause2012submodular}, and then (b) do ``pipage rounding''. I don't know how to express the part of the algorithm for (a); the part (b) is trivial in our case (if I am not mistaken), it is just a simple rounding of the continuous solution, since our matroid is a partition matroid~\cite{calinescu-max-submodular-matroid}.}

%\add{In fact, there is also another method that can give $\left(1-\frac{1}{e}\right)$-approximation~\cite{filmus2014monotone}, however, I did not understand it and could not express an algorithm based on it.}

\section{Femtocaching with Related Content Delivery}
\label{sec:second-utility-model}
We have so far considered a system corresponding to the example of Fig~\ref{fig:mobile-app-example}, where a cache-aware system \emph{recommends} alternative contents to users (in case of a cache miss), but users might not accept them. In this section, we consider a system closer to our second example of Fig~\ref{fig:mobile-app-example-second-model}, where the system \emph{delivers} some related content that is locally available \emph{instead of the original content}, in case of a cache miss. While a more extreme scenario, we believe this might still have application in a number of scenarios, as explained in Section~\ref{sec:intro} (e.g., for low rate plan users under congestion, or in limited access scenarios~\cite{internet-org,free-basics}). We are therefore interested to model and analyze such systems as well.
%we consider a variation of our basic model, which can describe (or better fit) some other real scenarios or future services\footnote{We do not claim that one of the two models is more realistic or feasible. The selection of the model depends on the future needs and evolution of mobile networks market and technology. Our goal here is to establish a generic framework that allows us to analyse the performance and propose solutions for a wide range of future scenarios.}. We define and analyse the performance of this new/different system model, and prove that our algorithms apply also to this case with the same performance guarantees. 
Due to space limitation, we present only the more generic femto-cache case of Section~\ref{sec:femto}; the analysis and results for the single cache cases of Section~\ref{sec:single} follow similarly.

%We assume system that does \textit{not} satisfy user requests for contents that are not locally cached, and only \rem{recommends or delivers} alternative contents to them; see Fig.~\ref{fig:mobile-app-example-second-model} for an example of a mobile app implementing this system. Such a system could be used in a number of scenarios, such as (i) rural or developing areas with limited deployment of mobile networks~\cite{KioskNet}, (ii) when access to only a few Internet services is provided, e.g., the Facebook's Internet.org  (Free Basics) project~\cite{internet-org,free-basics}, (iii) when the cellular network is temporarily overloaded, (iv) \textit{zero-rating} services~\cite{t-mobile-music-freedom,neutral-net-neutrality}, (v) etc. 

\begin{comment}
\begin{figure}[h]
\subfigure[]{\includegraphics[width=0.325\columnwidth]{figs/mobile_app_example_second_model_1.eps}\label{fig:mobile-app-example-second-model}}
\hspace{0.1\linewidth}
\subfigure[]{\includegraphics[width=0.33\columnwidth]{figs/mobile_app_example_second_model_2.eps}\label{fig:mobile-app-example-second-model-rating}}
\caption{Mobile app example for \textit{Alternative Content Delivery}: (a) delivering alternative content, and (b) rating the delivered content.}
\end{figure}
\end{comment}

Since original requests might not be served, the (soft) cache hit ratio metric does not describe sufficiently the performance of this system. To this end, we modify the definition of content utility:
\begin{definition}\label{def:utility-as-satisfaction}
When a user $i$ requests a content $k$ that is not locally available and the content provider delivers an alternative content $n$ then the user satisfaction is given by the utility $u_{kn}^{i}$. {$u_{kn}^{i} \in R$ is a real number, and does not denote a probability of acceptance, but rather the happiness of user $i$ when she receives $n$ instead of $k$. Furthermore $u_{kk}^{i} = U_{max}, \forall i$.}
\end{definition}
%\textit{Note}: we stress that the utilities $u_{kn}^{i}$ in Def.~\ref{def:utility-as-satisfaction} do not represent the probability a user $i$ to accept a content $n$ (as in Def.~\ref{def:utility-as-probability}), but the satisfaction of user $i$ given that she accepted content $n$. \rem{User satisfaction can be estimated by past statistics, or user feedback, e.g., by asking user to rate the received alternative content as in the example app of Fig.~\ref{fig:mobile-app-example-second-model-rating}.}

\textit{Note}: we stress that the utilities $u_{kn}^{i}$ in Def.~\ref{def:utility-as-satisfaction} do not represent the probability a user $i$ to accept a content $n$ (as in Def.~\ref{def:utility-as-probability}), but the satisfaction of user $i$ given that she accepted content $n$. {User satisfaction can be estimated by past statistics, or user feedback, e.g., by asking user to rate the received alternative content.}

Let us denote as $G_{i}(t)\subseteq \mathcal{M}$ the set of SCs with which the user $i$ is associated at time $t$. Given Def.~\ref{def:utility-as-satisfaction}, when a user $i$ requests at time $t$ a content $k$ that is not locally available, we assume a system (as in Fig.~\ref{fig:mobile-app-example-second-model}) that delivers to the user the cached content with the \textit{highest} utility\footnote{Equivalently, the system can recommend all the stored contents to the user and then allow the user to select the content that satisfies her more.}, i.e., the content $n$ where
\begin{equation}
n \equiv{\arg\max}_{\ell\in\mathcal{K}, j\in G_{i}(t)} ~\left\{u_{k\ell}^{i}\cdot x_{\ell j}\right\}
\end{equation}
Hence, the satisfaction of a user $i$ upon a request for content $k$ is
\begin{equation}
\max_{n\in\mathcal{K}, j\in G_{i}(t)}\left\{u_{kn}^{i} \cdot x_{nj}\right\}
\end{equation}
Using the above expression and proceeding similarly to Section~\ref{sec:femto}, it easily follows that the optimization problem that the network needs to solve to optimize the total user satisfaction in the system (among all users and all content requests), which we call \textit{soft cache hit user satisfaction} (SCH-US), is:
\begin{problem} \label{problem:femto-second-utility-model}
The optimal cache placement problem for the femtocaching scenario with alternative content delivery and content relations described by the matrix $\mathbf{U} = \{{u_{kn}^{i}}\}$, is
\begin{eqnarray}
\underset{X = \{x_{1}, ..., x_{K}\}}{\mbox{maximize }} \;  {f(X) = {\sum_{i=1}^{N}}\sum_{k=1}^{K} p_{k}^{i} \cdot \max_{n\in\mathcal{K}, j\in\mathcal{M}} \left(u_{kn}^{i} \cdot x_{nj}\cdot q_{ij}\right)}, \label{eq:objective-femto-second-model}\\
 \sum_{k =1}^{K} x_{kj} \leq C, ~~~\forall j \in \mathcal{M}.\label{eq:constraint-femto-second-model}
\end{eqnarray}
\end{problem}

For the sub-cases 1 and 2 of Def.~\ref{def:utility-as-probability} presented in Section~\ref{sec:model-sch}, the following corollary holds. 
\begin{corollary}\label{thm:corollary-femto-second-model-subcases}
The expression of \eq{eq:objective-femto-second-model} needs to be modified as%For the sub-cases 1 and 2 of Def.~\ref{def:utility-as-probability}, the expression of \eq{eq:objective-femto-second-model} needs to be modified as:
\begin{align}
\max_{n\in\mathcal{K}, j\in\mathcal{M}} \left(u_{kn}^{i} \cdot x_{nj} \cdot q_{ij}\right) 
	\rightarrow &q_{ij}\cdot E\left[\max_{n\in S}\{u_{kn}^{i}\}\right] =\nonumber\\
				&~~~= q_{ij}\cdot \int\left(1-\prod_{n\in S}F_{kn}(x)\right)dx \\
u_{kn}^{i} \rightarrow &u_{kn}
\end{align}
(where $S=\{(\ell,m):\ell\in\mathcal{K}, m\in\mathcal{M},x_{\ell m}=1\}$) for the sub-cases 1 and 2 of Def.~\ref{def:utility-as-probability}, respectively. 
\end{corollary}
\begin{proof}
The proof is given in Appendix~\ref{app:corollary-subcases-second-model}.
\end{proof}

%\underline{Remark}: For the sub-cases 1 and 2 of Def.~\ref{def:utility-as-probability} presented in Section~\ref{sec:model-sch}, it can be shown  that in the expression of the \eq{eq:objective-femto-second-model} for the SCH-US, we just need to substitute (for sub-case 1)
%\begin{align}
%\max_{n\in\mathcal{K}\add{, j\in\mathcal{M}}} \left(u_{kn}^{i} \cdot x_{nj} \add{\cdot q_{ij}}\right) ~~\rightarrow~~ \add{q_{ij}\cdot} E\left[\max_{n\in S}\{u_{kn}^{i}\}\right]\equiv\int\left(1-\prod_{n\in S}F_{kn}(x)\right)dx
%\end{align}
%where $S=\{\ell:\ell\in\mathcal{K},x_{\ell}=1\}$, or (for sub-case 2)
%\begin{align}
%u_{kn}^{i} ~~\rightarrow~~ u_{kn}
%\end{align}

We now prove the following Lemma, which shows that Theorem~\ref{thm:greedy} and Algorithm~\ref{alg:greedy-single-cache} apply also to the Optimization Problem~\ref{problem:femto-second-utility-model}.

\begin{lemma}\label{lemma:np-and-submodular-second-model}~\\
(1) The Optimization Problem~\ref{problem:femto-second-utility-model} is NP-hard, \\
(2) with submodular and monotone objective function (\eq{eq:objective-femto-second-model}).%.
\end{lemma}
\begin{proof} The proof is given in Appendix~\ref{app:lemma-second-model}.
\end{proof}

\section{Performance Evaluation}
\label{sec:sims}
%In this section, we perform MATLAB simulations in order to evaluate the effects of soft cache hits for single cache and femto-caching. We compare our results with the traditional femto-caching framework, and we show in what scenarios soft cache hits can significantly improve the cache hit ratio.

\subsection{Simulation setup}
\textbf{Contents dataset}. We consider a real dataset of YouTube videos from~\cite{youtube-related-videos-dataset}. The dataset contains several information about the videos, such as their (a) popularity and (b) size, as well as (c) the list of related videos (as recommended by YouTube) for each of them. This information allows us to simulate scenarios with real parameters $p_{k}^{i}$, $s_{k}$, and $u_{kn}^{i}$. After pre-processing the data to remove contents with no popularity values, we build the related content matrix (utility matrix $\textbf{U}$).  Due to the sparsity of the dataset, we only consider contents belonging to the largest connected component, that includes $K=2098$ videos. The average number of related content for these videos is $3.6$. Since our dataset does not contain per-user information, we consider the sub-case-2 of Definition~\ref{def:utility-as-probability}, and assume that if content $k$ is related to content $n$ in the dataset, then $u_{kn} = 1$. However, we later perform a sensitivity analysis as a function of diminishing acceptance probabilities for related content.

\textbf{Cellular network}. We consider an area of 1 km$^2$ that contains $M=20$ SCs. SCs are randomly placed (i.e., uniformly) in the area, which is an assumption that has been also used in similar works~\cite{femto,poularakis2014}. An SC can serve a request from a user, when the user is inside its communication range, which we set to $200$ meters. We also consider $N=50$ mobile users. This creates a relatively dense network, where a random user is connected to $3$ SCs \emph{on average}. We will also consider sparser and denser scenarios, for comparison. We generate a set of 20~000 requests according to the content popularity calculated from the UouTube dataset, over which we average our results.

%\emph{Simulator}. We build a MATLAB simulator that works as follows. First, it computes the content to allocate for each SC: in the scenario of a single cache with soft cache hits, using Algorithm~\ref{alg:greedy-single-cache} or Algorithm~\ref{alg:fast-single-knapsack}; in the femtocaching case with soft cache hits using Algorithm 4~\ref{alg:greedy-femto}; finally, for the femto-caching without soft cache hits (our baseline).  For each request, the simulator verifies if either the content can be obtained from a SC inside the communication range (cache hit) or a related content exists (SCH); otherwise, there is a cache miss.

Unless otherwise stated, the simulations use the parameters summarized in Table~\ref{tab:parameter}.

\begin{table}[h]
\centering
\caption{Parameters used in the simulations.}
\begin{small}
\begin{tabular}{|c|c|c|c|}
\hline
\textbf{Parameter} & \textbf{Value} & \textbf{Parameter} & \textbf{Value} \\ \hline
nb. of contents, $K$	& 2098	& nb. of requests 	& 20.000 \\ \hline
Cache size, $C$			& 5		& nb. of SCs,	$M$		& 20  \\  \hline
Area					& 1	x 1 km	& Communication range	& 200 m	  \\ \hline
\end{tabular}
\end{small}
\label{tab:parameter}
\end{table}
%
%\begin{table}
%\centering
%\caption{Parameters used in the simulations.}
%\begin{tabular}{|>{\centering\arraybackslash}m{0.95in}|>{\centering\arraybackslash}m{0.4in}|>{\centering\arraybackslash}m{0.95in}|>{\centering\arraybackslash}m{0.4in}|@{}m{0pt}@{}}
%\hline
%\textbf{Parameter} & \textbf{Value} & \textbf{Parameter} & \textbf{Value} & \\ [1.2ex] \hline
%Number of contents $K$	& 2098	& Number of requests 	& 20.000 & \\ [1.2ex] \hline
%Cache size $C$			& 5		& Number of SCs	$M$		& 20 & \\ [1.2ex] \hline
%Area					& 1	x 1 km	& Communication range	& 200 m	&  \\ [1.2ex] \hline
%\end{tabular}
%\label{tab:parameter}
%\end{table}

\subsection{Performance Results}

We consider the following four content caching schemes:
\begin{itemize}
\item \emph{Single (popularity-based)}: Single cache accessible per user (e.g., the closest one). Only normal cache hits allowed, and the most popular contents are stored in each cache. This is the baseline scheme, commonly used in related works.
\item \emph{SingleSCH}: Single cache with soft cache hits, with the content allocation given by Algorithm~\ref{alg:greedy-single-cache} (or Algorithm~\ref{alg:fast-single-knapsack}).
\item \emph{Femto}: Femto-caching without soft cache hits {(from~\cite{femto})}.
\item \emph{FemtoSCH}: Femto-caching with soft cache hits, with the content allocation given by Algorithm~\ref{alg:greedy-femto}.
\end{itemize}

\textbf{Cache size impact:} We first investigate the impact of cache size, assuming fixed content sizes. Fig.~\ref{fig:sens_cache_size} depicts the total cache hit ratio, for different cache sizes $C$: we consider a cache size per SC between 2 and 15 contents. The simulations suggest that soft cache hits (SCH) can double the cache hit ratio. What is more, these gains are applicable to both the single cache and femto-caching scenarios, which show that our approach can offer considerable benefits \emph{on top of femto-caching}, which as we see can already achieve an improvement of more than $50\%$, compared to single caches in this scenario. The two methods together offer a total of $3\times$ improvement compared to the baseline scenario ``Single'', reaching a maximum cache hit ratio of about $60\%$ for $C=15$. Finally, even with a cache size of per SC of about $0.1\%$ of the total catalog, introducing soft cache hits offers $30\%$ cache hit ratio, which is promising.

\begin{figure}\centering
		\includegraphics[width=0.8\columnwidth]{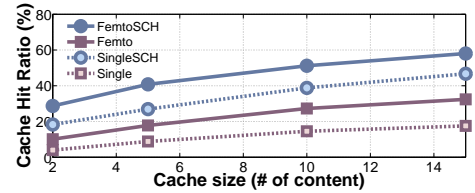}
		\caption{Cache hit ratio vs. $C$, with fixed content size.}
		\label{fig:sens_cache_size}
\end{figure}

\textbf{Variable file size:} In Fig.~\ref{fig:sens_cache_size_alg2}, we now also take into account the different file sizes when optimizing our allocation (these are available in the YouTube dataset).  Comparing the performance of even this less theoretically efficient greedy algorithm (Algorithm~\ref{alg:fast-single-knapsack}) to the single cache with no soft hits, already reveals considerable performance improvement. In fact, Algorithm~\ref{alg:fast-single-knapsack} exploits the fact that contents have different sizes, to ``replace'' longer contents with related ones that might be shorter. While one could of course not replace a very large content with a very small one, we have observed in our dataset that the sizes of related contents are not independent (i.e., related videos of large videos are indeed large, and vice versa, but have enough difference that can sometimes be exploited by the size-aware algorithm). 

%While with a synthetic utility matrix (e.g., random) simulations might be biased (since Algorithm 2 would only choose small content leading to much higher gains), we use real data that guarantees user satisfaction in a realistic scenario\footnote{.}.

\begin{figure}[h]
		\includegraphics[width=0.8\columnwidth]{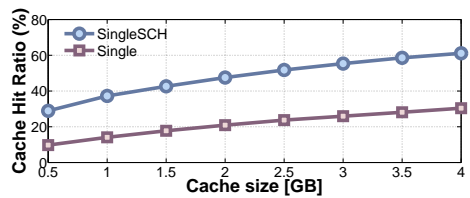}
		\caption{Cache hit ratio vs. $C$ with variable content size.}
		\label{fig:sens_cache_size_alg2}
\end{figure}

\textbf{SC density impact:} In Fig.~\ref{fig:sens_small_cells} we perform a sensitivity analysis with respect to the number of SCs in the area (assuming fixed capacity $C=5$). We test 2 sparse scenarios ($M=5$ and $M=10$) and 2 dense scenarios ($M=20$ and $M=30$). The average number of SCs that can be seen by a user varies from around $1$ ($M=5$) to $4.6$ ($M=30$). In the sparse scenarios, a user can usually see at most one SC. For this reason, ``Femto'' and ``Single'' perform similarly ($20-30$\% cache hit rate). As the SC density increases, the basic femto-caching is able to improve performance, as expected. However, femtocaching with SCH brings even more performance gains. With a storage capacity per SC of about $0.25\%$ of the content catalog ($5$ out of $2000$ contents), and a coverage overlap of 2-4 SCs per user, femto-caching together with SCH can achieve a $30-50\%$ cache hit ratio. This is promising on the additive gains of the two methodologies.

    \begin{figure}[h]
		\includegraphics[width=0.8\columnwidth]{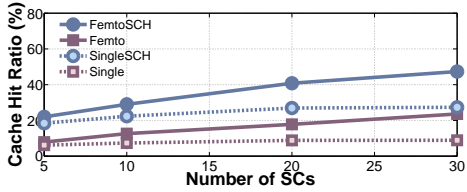}
		\caption{Cache hit ratio vs. number of SCs $M$ ($C=5$).}
		\label{fig:sens_small_cells}
    \end{figure}

\textbf{Utility matrix impact:} In these final two sets of simulations, we further investigate the impact of the content relations as captured by the matrix $\mathbf{U}$ (and its structure). A first important parameter to consider is the average number of related contents per video, which we denote as $E[R]$, where the set $R$ for a content $k$ is defined as $R = \{n\in\mathcal{K}: u_{kn}>0\}$. In the previous scenarios, the number of related contents was inferred from the YouTube trace, and was found to be equal to $E[R] = 3.6$. To understand the impact of this parameter, in this next scenario we generate two synthetic content graphs $\mathbf{U}$:
\begin{itemize} 
\item  \emph{SCH(1)}: content $k$ picks content $n$ as related content with probability $p_n$ (i.e., proportional to its popularity), normalized to a mean $E[R]$ value per content.
\item \emph{SCH(2)}: any content picks $E[R]$ related content randomly chosen.
\end{itemize}
While the latter assumes that content relations are independent of their popularity, the former assumes that a more popular content has a higher chance to appear in the related content list. In fact, this is quite inline with daily experience of how recommendation systems work.

In Fig.~\ref{fig:sens_utility}, we compare the cache hit ratio for single and femto-caching scenarios: without SCH, with SCH1, and with SCH2, assuming that $\mathbf{E}[R]$ varies between $2$ and $10$ related content items. A first observation is that, due to the sparsity of the content matrix, SCH(2) (i.e., random content relations) brings only marginal improvements to the total number of hits. On the other hand, a correlation between related content and popularity (i.e., SCH(1)) is what brings considerable offloading gains, even for small $E[R]$. In fact, comparing these synthetic results with the previous trace-based ones, one can infer that the real dataset probably more closely resembles SCH1, i.e. does exhibit such a correlation.

\begin{figure}[h]
\includegraphics[width=0.8\columnwidth]{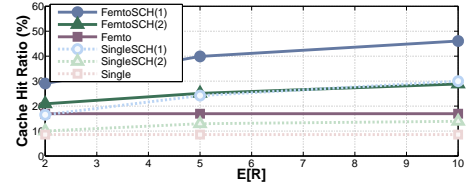}
\caption{Cache hit ratio for different number of related contents $E[R]$; synthetic traces.}
\label{fig:sens_utility}
\end{figure}

\textbf{User flexibility:} In this last scenario, we present in Fig.\ref{fig:ukn} the cache hit ratio as a function of the willingness of a user to accept a related content. We consider two scenarios, both in the femto-caching context:

%Finally, in this last scenario, we consider generic (Type 2) utilities, where a related content is associated with a ``utility'' $u_{ij}$ that captures the probability of this related content being accepted by the user, if the original content $i$ is not available and content $j$ is recommended instead. In Fig.\ref{fig:ukn}, we depict the cache hit ratio as a function of the willingness of a user to accept a related content. We consider two scenarios, both in the femto-caching context: 
\begin{itemize}
\item \emph{Synthetic:} We generate a synthetic matrix $U$ as in SCH(1) above, with $E[R] = 4$ and $u_{kn} = u < 1$.
\item \emph{YouTube:} We use the real YouTube dataset for the matrix $U$. However, all related contents also have a utility $u_{kn} = u < 1$ (instead of $u_{kn} = 1$ considered in the previous scenarios).
\end{itemize}
On the x-axis of Fig.\ref{fig:ukn} we vary this parameter $u$ from $0$ to $1$. As can be seen there, when the user's acceptance probability becomes very small, the scenario becomes equivalent to standard femto-caching without soft hits, and the gains reported there are inline with the previous plots. However, as user willingness to accept related content increases, the optimization policy can exploit opportunities for potential soft cache hits and improve performance. E.g. for a probability $50\%$ to accept an alternative recommended content, cache hit ratios increase by almost $2 \times$ (from $15\%$ to $27\%$ in the YouTube dataset). Results are in fact very comparable for the synthetic and YouTube traces. Finally, we observed similar behavior in scenarios conforming to the model of Section~\ref{sec:second-utility-model}, where $u_{kn}$ do not denote probabilities of acceptance, but correspond to the user satisfaction.

    \begin{figure}[h]
		\includegraphics[width=0.8\columnwidth]{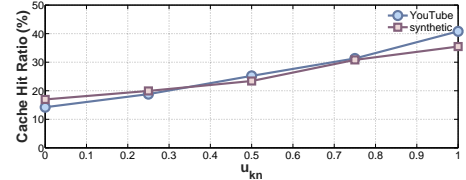}
		\caption{Cache hit ratio as a function of user's willingness to accept related content ($M = 20, C = 5, E[R] = 4$).}
		\label{fig:ukn}
    \end{figure}

%\begin{figure}
%    \centering
%    \begin{subfigure}[b]{0.22\textwidth}
%        \includegraphics[width=\textwidth]{figs/sch1.eps}
%        \caption{M = 5, C = 2.}
%        \label{fig:sch1}
%    \end{subfigure}
%    ~ %add desired spacing between images, e. g. ~, \quad, \qquad, \hfill etc. 
%      %(or a blank line to force the subfigure onto a new line)
%    \begin{subfigure}[b]{0.22\textwidth}
%        \includegraphics[width=\textwidth]{figs/sch2.eps}
%        \caption{M = 10, C = 5.}
%        \label{fig:sch2}
%    \end{subfigure}
%    \caption{Offloading gain.}\label{fig:buffer}
%\end{figure}
%
%\begin{figure}
%\includegraphics[width=\columnwidth]{figs/sens_utility2.eps}
%\caption{Sensitivity analysis - utility (details).}
%\label{fig:sens_utility2}
%\end{figure}

\section{Related Work}
\label{sec:related}
\textbf{Mobile Edge Caching.} Densification of cellular networks, overlaying the standard macro-cell network with a large number of small cells (e.g., pico- or femto-cells), has been extensively studied and is considered a promising solution to cope with data demand~\cite{hetnets-commag-2012,femtocells-survey-jsac-2012,hetnets-paradigm-shift}. As this densification puts a tremendous pressure on the backhaul network, researchers have suggested storing popular content at the ``edge'', e.g., at small cells~\cite{femto}, user devices~\cite{femtoD2D,Hui-offloading,sermpezis2014}, or vehicles acting as mobile relays~\cite{vigneri2016}. 

Our work is complementary to these approaches, as it can utilize such mobile edge caching systems while showing how to further optimize the cache allocation when there is a cache-aware recommender systems in place. We have applied this approach in the context of mobile (ad-hoc) networks with delayed content delivery~\cite{sch-chants-2016} as well, and applied it here for the first time in the context of femto-caching~\cite{femto}. Additional research directions have also recently emerged, more closely considering the interplay between caching and the physical layer such as Coded Caching~\cite{Ali2014} and caching for coordinated (CoMP) transmission~\cite{Psounis2015, Lau:PhyCache}. We believe the idea of soft cache hits could be applied in these settings as well, and we plan to explore this as future work.

\textbf{Caching and Recommendation Interplay.} Although not in the area of wireless systems, there exist some recent works that have jointly considered caching and recommendation in peer-to-peer networks~\cite{content-recommendation-swarming} and CDNs~\cite{what-should-you-cache-nossdav,cache-centric-video-recommendation}. Specifically, ~\cite{content-recommendation-swarming} studies the interplay between a recommendation system and the performance of content distribution on a peer-to-peer network like BitTorrent. The authors model and analyse the performance, and propose heuristics (e.g., based on the number of cached copies, or ``seeders'' in BitTorrent) for tuning the recommendation system, in order to improve performance and reduce content distribution costs. Although in our work we follow the opposite approach (given a recommendation system and user utilities, we optimize the caching algorithm), it would be interesting to investigate such recommendation system optimizations in the context of cellular networks.

~\cite{what-should-you-cache-nossdav} shows that users tend to follow YouTube's suggestions, and despite the large catalog of YouTube, the top-10 recommendations are usually common for different users in the same geographical region. Hence, CDNs can use the knowledge from the recommendation system to improve their content delivery (e.g., by storing in central caches the most appropriate contents). Finally, in~\cite{cache-centric-video-recommendation} the authors suggest a recommended list reordering approach that can achieve higher cache hit ratios (for YouTube servers/caches). In particular, they show that performance can be improved when the positions of contents in the related list of YouTube are not changed compared to the previous recommended list. The increasing dependence of user requests on the output of recommender systems clearly suggests that there is an opportunity to further improve the performance of (mobile) edge caching by jointly optimizing both, with minimum impact on user QoE.

%%%%%% some copy-paste text from Infocom that contains references%%%%%

%
%This suggests that even though studies assuming a large (CDN-type) cache deep inside the core network give promising hit ratios~\cite{Erman2011}
%
%More than 1TB storage would be needed to just store $0.1\%$ of either one~\cite{tutorial-sigmetrics}. Considering other VoD platforms, YouTube, etc., this number explodes. Even, with a skewed popularity distribution, \emph{local} cache hit ratios would be rather low~\cite{Paschos-misconceptions}. 
%
%
%
%While our problem formulation and analysis takes place in the context of \emph{femto-caching}, the idea of soft cache hits is also applicable to some of the other caching frameworks mentioned earlier. These include caching for opportunistic or delay-tolerant networking~\cite{Hui-offloading,Whitbeck-offloading,vigneri2016}, PHY-aware caching~\cite{Psounis2015, Lau:PhyCache}, or even wired caching, where soft cache hits could offer additional gain factors. 

%\section{Discussion}
%\label{sec:discussion}
%\input{discussion}

\section{Conclusions}
\label{sec:conclusions}
In this paper, we have proposed the idea of \emph{soft cache hits}, where an alternative content can be recommended to a user, when the one she requested is not available in the local cache. While normal caching systems would declare a cache miss in that case, we argue that an appropriate recommended content, related to the original one can still satisfy the user with high enough probability. We then used this idea to design such a system around femto-caching, and demonstrated that considerable additional gains, \emph{on top of those of femto-caching} can be achieved using realistic scenarios and data. We believe this concept of soft cache hits has wider applicability in various caching systems.

\bibliographystyle{ieeetr}%{abbrv}

\appendices
\section{Proof of Corollary~\lowercase{\ref{thm:corollary-single-cache-subcases}}}\label{app:corollary-single-subcases}
\textbf{Sub-case 1.} Since the exact per-user utilities are not known, we calculate the SCHR given in \eq{eq:schr-single-defintion} by taking the conditional expectations on $F_{kn}(x)$. Denoting the corresponding pdf as $f_{kn}(x)$, we proceed as follows:

\begin{footnotesize}
\begin{align*}
SCHR 
&= \sum_{i=1}^{N}\sum_{k=1}^{K} p_{k}^{i}\cdot q_{i} \cdot E\left[ \left(1 - \prod_{n=1}^{K} \left(1- u_{kn}^{i} \cdot x_{n}\right)\right)\right] \\
&= \sum_{i=1}^{N}\sum_{k=1}^{K} p_{k}^{i}\cdot q_{i} \cdot \left(1 -  E\left[\prod_{n=1}^{K} \left(1- u_{kn}^{i} \cdot x_{n}\right)\right] \right)\\
&= \sum_{i=1}^{N}\sum_{k=1}^{K} p_{k}^{i}\cdot q_{i} \cdot \left(1 -  \prod_{n=1}^{K} E\left[\left(1- u_{kn}^{i} \cdot x_{n}\right]\right) \right)\\
&= \sum_{i=1}^{N}\sum_{k=1}^{K} p_{k}^{i}\cdot q_{i} \cdot \left(1 - \prod_{n=1}^{K} \int \left(1- t \cdot x_{n}\right)\cdot f_{kn}(t)dt \right)\\
&= \sum_{i=1}^{N}\sum_{k=1}^{K} p_{k}^{i}\cdot q_{i} \cdot \left(1 - \prod_{n=1}^{K}  \left(1- \left(\int t\cdot f_{kn}(t) dt\right) \cdot x_{n}\right) \right)\\
&= \sum_{i=1}^{N}\sum_{k=1}^{K} p_{k}^{i}\cdot q_{i} \cdot \left(1 - \prod_{n=1}^{K}  \left(1- E[u_{kn}^{i}] \cdot x_{n}\right) \right)
\end{align*}
\end{footnotesize}
where (i) the third equation holds since the utilities for different content pairs \{k,n\} are independent, and thus the expectation of their product is equal to the product of their expectations, and (ii) we denoted 
\begin{align*}
E[u_{kn}^{i}]&\equiv\int t\cdot f_{kn}(t) dt=\int (1-F_{kn}(t)) dt
\end{align*}
and the above equation holds since $u_{kn}^{i}$ is a positive random variable.

\textbf{Sub-case 2} follows straightforwardly.

\section{Proof of Lemma~\lowercase{\ref{lemma:single-hardness}}}\label{app:single-hardness}
We prove here the NP-hardness of the optimal cache allocation for a single cache with soft cache hits. Let us consider an instance of Optimization Problem~\ref{problem:single-cache}, where the utilities are equal among all users and can be either 1 or 0, i.e., $u_{kn}^{i} = u_{kn},~\forall i\in\mathcal{N}$ and $u_{kn}\in\{0,1\},~\forall k,n\in\mathcal{K}$. We denote as $\mathcal{R}_{k}$ the set of contents related to content $k$, i.e.
\begin{equation}\label{eq:related-set}
 \mathcal{R}_{k} =\{n \in \mathcal{K}: n \ne k, u_{kn} > 0\} \;\; \mbox{(related content set)}
\end{equation}

 %Proving NP-hardness for a sub-case, implies that the general problem is NP-hard as well.
%Consider the content subsets $\mathcal{S}_{k} = \{k\} \cup \mathcal{R}_{k}$, $k \in \mathcal{K}$, where we remind the reader that $\mathcal{R}_{k}$ is the set of contents related to $k$ (see Section~\ref{}).

%  and the collection of sets $S=\{S_{1}, S_{2}, ..., S_{k}\}$, where each subset $S_{i} = \left\{ n \in\mathcal{K} : u_{kn}=1\right\}\subseteq \mathcal{K}$ comprises the related contents of a content $i$ (including the content $i$).

Consider the content subsets $\mathcal{S}_{k} = \{k\} \cup \mathcal{R}_{k}$. Assume that only content $k$ is stored in the cache ($x_{k}=1$ and $x_{n} = 0, \forall n \ne k$). All requests for contents in $S_{k}$ will be satisfied (i.e. ``covered'' by content $k$), and thus SCHR will be equal to $\sum_{i\in\mathcal{N}}\sum_{n \in S_{k}}p_{n}^{i}\cdot q_{i}$. When more than one contents are stored in the cache, let $\mathcal{S^{'}}$ denote the union of all contents covered by the stored ones, i.e., $\mathcal{S^{'}} = \bigcup_{ \{k: x_{k} = 1\} } S_{k}$. Then, the SCHR will be equal to $\sum_{i\in\mathcal{N}}\sum_{n \in S^{'}}p_{n}^{i}\cdot q_{i}$. Hence, the Optimization Problem~\ref{problem:single-cache} becomes equivalent to
\begin{align*}
\max_{\mathcal{S^{`}}} ~~\sum_{n \in \mathcal{S^{'}}}p_{n}^{i}\cdot q_{i} ~~~~~~~~~~s.t.~~\left| \{k: x_{k} = 1\} \right| \leq C.
\end{align*} 
This corresponds to the the \textit{maximum coverage problem with weighted elements}, where ``elements'' (to be ``covered'') correspond to the contents $i\in\mathcal{K}$, weights correspond to the probability values $p_{n}^{i}\cdot q_{i}$, the number of selected subsets $\{k: x_{k} = 1\}$ must be less than $C$, and their union of covered elements is $\mathcal{S^{'}}$. This problem is known to be a NP-hard problem~\cite{budget-max-cover}, and thus the more generic problem (with different $u_{kn}^{i}$ and  $0\leq u_{kn} \le 1$) is also NP-hard.  
%
%by reduction to a variant of (partial) set cover. namely, we have the set of all contents K, and subsets corresponding to the rows of $\mathbf{U}$. e.g. the subject $\mathcal{K}_{i} = \{n \in \mathcal{K}: u_{kn} = 1\}$. We want to pick at most $C$ of these subsets in order to maximize the sum of values of elements ``covered'' by these $C$ subsets. 

\section{Proof of Lemma~\lowercase{\ref{lemma:single-submodular}}}\label{app:single-submodular}
The objective function of \eq{eq:objective-single-cache} $f(X):\{0,1\}^{K}\rightarrow \mathbb{R}$ is equivalent to a set function $f(S):2^{\mathcal{K}}\rightarrow \mathbb{R}$, where $\mathcal{K}$ is the finite \emph{ground set} of contents, and $S = \{k \in \mathcal{K}: x_{k}=1\}$. In other words, 
\begin{equation}
f(S)\equiv \sum_{i=1}^{N}\sum_{k=1}^{K} p_{k}^{i}\cdot q_{i} \cdot \left(1-  \prod_{n\in S}\left(1-u_{kn}^{i}\right) \right).
\end{equation}

A set function is characterised as \textit{submodular} if and only if for every $A\subseteq B\subset V$ and $\ell\in V\backslash B$ it holds that
\begin{equation}
\left[f\left(A\cup \{\ell\} \right) - f\left(A \right)\right] - \left[f\left(B\cup \{\ell\} \right) - f\left(B \right)\right] \geq 0
\end{equation}

From \eq{eq:objective-single-cache}, we first calculate 
\begin{footnotesize}
\begin{align*}
& f\left(A\cup \{\ell\} \right) - f\left(A \right) =\\
& \hspace{-0.1cm}=	\sum_{i=1}^{N}\sum_{k=1}^{K} p_{k}^{i} q_{i} \left(1 - \hspace{-0.2cm}\prod_{n \in A\cup \{\ell\}} \hspace{-0.2cm}\left(1- u_{kn}^{i} \right)\right) \\
&\hspace{1.5cm}- \sum_{i=1}^{N}\sum_{k=1}^{K} p_{k}^{i} q_{i}  \left(1 - \prod_{n \in A} \left(1-  u_{kn}^{i} \right)\right) \\
& \hspace{-0.1cm}=	\sum_{i=1}^{N}\sum_{k=1}^{K} p_{k}^{i}\cdot q_{i} \cdot \left( \prod_{n \in A} \left(1- u_{kn}^{i} \right) - \prod_{n \in A\cup \{\ell\}} \left(1- u_{kn}^{i} \right)\right) \\ 
& \hspace{-0.1cm}=	\sum_{i=1}^{N}\sum_{k=1}^{K} p_{k}^{i}\cdot q_{i} \cdot \left(   u_{k\ell}^{i}  \cdot  \prod_{n \in A} \left(1- u_{kn}^{i} \right) \right).
\end{align*}
\end{footnotesize}
%
%\begin{small}
%\begin{align*}
%& f\left(A\cup \{\ell\} \right) - f\left(A \right) =\\
%&	\sum_{k=1}^{K} p_{k} \cdot \left(1 - \prod_{n \in A\cup \{\ell\}} \left(1- u_{kn} \cdot x_{n}\right)\right) - \sum_{k=1}^{K} p_{k} \cdot \left(1 - \prod_{n \in A} \left(1-  u_{kn} \cdot x_{n}\right)\right) =\\
%&	\sum_{k=1}^{K} p_{k} \cdot \left( \prod_{n \in A} \left(1- u_{kn} \cdot x_{n}\right) - \prod_{n \in A\cup \{\ell\}} \left(1- u_{kn} \cdot x_{n}\right)\right)= \\ 
%&	\sum_{k=1}^{K} p_{k} \cdot \left(   u_{k\ell} \cdot x_{\ell} \cdot  \prod_{n \in A} \left(1- u_{kn} \cdot x_{n}\right) \right).
%\end{align*}
%\end{small}
Then,
\begin{footnotesize}
\begin{align*}
&\left[f\left(A\cup \{\ell\} \right) - f\left(A \right)\right] - \left[f\left(B\cup \{\ell\} \right) - f\left(B \right)\right] =\\
&={\sum_{i=1}^{N}}\sum_{k=1}^{K} {p_{k}^{i} q_{i}}  \left(  {u_{k\ell}^{i}}   \prod_{n \in A} \left(1- {u_{kn}^{i}} \right) \right)  \\
&\hspace{1.5cm}
- {\sum_{i=1}^{N}}\sum_{k=1}^{K} {p_{k}^{i} q_{i}}  \left( {u_{k\ell}^{i}}   \prod_{n \in B} \left(1- {u_{kn}^{i}} \right) \right) \\
&={\sum_{i=1}^{N}}\sum_{k=1}^{K} {p_{k}^{i} q_{i}} \cdot {u_{k\ell}^{i}} \cdot  \left(\prod_{n \in A} \left(1- {u_{kn}^{i}} \right) - \prod_{n \in B} \left(1- {u_{kn}^{i}} \right) \right)\\
&={\sum_{i=1}^{N}}\sum_{k=1}^{K} {p_{k}^{i} q_{i}} \cdot {u_{k\ell}^{i}} \cdot  \prod_{n \in A} \left(1-  {u_{kn}^{i}} \right) \cdot \left(1 - \prod_{n \in B\backslash A} \left(1- {u_{kn}^{i}} \right) \right)
\end{align*}
\end{footnotesize}
 
%\begin{small}
%\begin{align*}
%&\left[f\left(A\cup \{\ell\} \right) - f\left(A \right)\right] - \left[f\left(B\cup \{\ell\} \right) - f\left(B \right)\right] =\\
%&\sum_{k=1}^{K} p_{k} \cdot \left(  u_{k\ell} \cdot x_{\ell} \cdot  \prod_{n \in A} \left(1- u_{kn} \cdot x_{n}\right) \right) \\
%&- \sum_{k=1}^{K} p_{k} \cdot \left( u_{k\ell} \cdot x_{\ell} \cdot  \prod_{n \in B} \left(1- u_{kn} \cdot x_{n} \right) \right) = \\
%&\sum_{k=1}^{K} p_{k} \cdot u_{k\ell} \cdot x_{\ell} \cdot  \left(\prod_{n \in A} \left(1- u_{kn} \cdot x_{n}\right) - \prod_{n \in B} \left(1- u_{kn} \cdot x_{n}\right) \right)=\\
%&\sum_{k=1}^{K} p_{k} \cdot u_{k\ell} \cdot x_{\ell} \cdot  \prod_{n \in A} \left(1-  u_{kn} \cdot x_{n}\right) \cdot \left(1 - \prod_{n \in B\backslash A} \left(1- u_{kn} \cdot x_{n}\right) \right)
%\end{align*}
%\end{small}

The above expression is always $\geq 0$, which proves the submodularity for function $f$.

Furthermore, the function $f$ is characterised as \textit{monotone} if and only if $f(B)\geq f(A)$ for every $A\subseteq B\subset V$. In our case, this property is shown as

\begin{footnotesize}
\begin{align*}
&f(B)-f(A)  =\\
&\hspace{-0.1cm}= {\sum_{i=1}^{N}}\sum_{k=1}^{K} {p_{k}^{i} q_{i}} \cdot \left(1 - \prod_{n \in B} \left(1- {u_{kn}^{i}} \right)\right)  \\
&\hspace{1.5cm}- {\sum_{i=1}^{N}}\sum_{k=1}^{K} {p_{k}^{i} q_{i}} \cdot \left(1 - \prod_{n \in A} \left(1- {u_{kn}^{i}} \right)\right) \\
&\hspace{-0.1cm}= {\sum_{i=1}^{N}}\sum_{k=1}^{K} {p_{k}^{i} q_{i}} \cdot \left(\prod_{n \in A} \left(1- {u_{kn}^{i}} \right) - \prod_{n \in B} \left(1- {u_{kn}^{i}} \right)\right) \\
&\hspace{-0.1cm}= {\sum_{i=1}^{N}}\sum_{k=1}^{K} {p_{k}^{i} q_{i}} \cdot \prod_{n \in A} \left(1- {u_{kn}^{i}} \right) \cdot \left(1 - \prod_{n \in B\backslash A} \left(1- {u_{kn}^{i}} \right)\right) \geq 0
\end{align*}
\end{footnotesize}
%\begin{small}
%\begin{align*}
%&f(B)-f(A)  =\\
%& \sum_{k=1}^{K} p_{k} \cdot \left(1 - \prod_{n \in B} \left(1- u_{kn} \cdot x_{n}\right)\right) - \sum_{k=1}^{K} p_{k} \cdot \left(1 - \prod_{n \in A} \left(1- u_{kn} \cdot  x_{n}\right)\right) =\\
%& \sum_{k=1}^{K} p_{k} \cdot \left(\prod_{n \in A} \left(1- u_{kn} \cdot x_{n}\right) - \prod_{n \in B} \left(1- u_{kn} \cdot x_{n}\right)\right) =\\
%& \sum_{k=1}^{K} p_{k} \cdot \prod_{n \in A} \left(1- u_{kn} \cdot x_{n}\right) \cdot \left(1 - \prod_{n \in B} \left(1- u_{kn} \cdot x_{n}\right)\right) \geq 0
%\end{align*}
%\end{small}

\section{Proof of Theorem~\lowercase{\ref{lemma:single-knapsack}}}\label{app:single-knapsack}
Following similar arguments as in the proof of Lemma~\ref{lemma:single-hardness}, the Optimization Problem~\ref{problem:single-knapsack} can be shown to be equivalent to the \textit{budgeted maximum coverage problem with weighted elements}, which is an NP-hard problem~\cite{budget-max-cover}. 

In Algorithm~\ref{alg:fast-single-knapsack}, we first calculate a solution $S^{(1)}$ returned by a modified version (\textproc{ModifiedGreedy}) of the greedy algorithm (line~1). The differences between the greedy algorithm (e.g., Algorithm~\ref{alg:greedy-single-cache}) and \textproc{ModifiedGreedy}, are that the latter: (a) each time selects to add in the cache the content that increases the most the fraction of the objective function over its own size (line~13), and (b) considers every content, until there is no content that can fit in the cache (lines~14-20). Then, Algorithm~\ref{alg:fast-single-knapsack} calculates the solution $S^{(2)}$ that the greedy algorithm would return if all contents were of equal size (line~2). The returned solution, is the one between $S^{(1)}$ and $S^{(2)}$ that achieves a higher value of the objective function (lines~3-7). 

Hence, Algorithm~\ref{alg:fast-single-knapsack} is a ``fast-greedy'' type of approximation algorithm. Since, the objective function was shown to be submodular and monotone in Lemma~\ref{lemma:single-submodular}, our fast greedy approximation algorithm can achieve a $\frac{1}{2}\cdot \left(1-\frac{1}{e}\right)$-approximation solution (in the worst case), when there is a Knapsack constraint, using similar arguments as in~\cite{leskovec-max-submodular-knapsack}. 
%\add{The more computationally complex proposed Algorithm~\ref{alg:complex-single-knapsack} is based on~\cite{max-submodular-knapsack}, and was shown to achieve (in the worst case) a $\cdot \left(1-\frac{1}{e}\right)$-approximation.}

\section{Proof of Lemma~\lowercase{\ref{lemma:femto-np-matroid}}}\label{app:femto-NP-matroid}
\underline{Item (1):} In the single cache case, we reduced the Optimization Problem~\ref{problem:single-cache} to a weighted maximum coverage problem, where a set of contents need to be selected (i.e., to be stored) in order to maximize the weights (i.e., probabilities $p_{k}^{i}\cdot q_{i}$) of the ``elements'' (i.e., other contents) with which they are connected (i.e., edges/utilities $u_{kn}^{i}>0$). The ground set of contents was $\mathcal{K}$, and the ground set of ``elements'' was $\mathcal{K}$ as well.

In the case of multiple users and overlapping caches, the Optimization Problem~\ref{problem:femto-related} can similarly be reduced to the NP-hard weighted maximum coverage problem, if (i) instead of the set of contents $\mathcal{K}$, we consider the set of tuples \{content,user\} $\mathcal{K}\times \mathcal{N}$, and (ii) instead of ``elements''/contents, we consider the set of ``elements''/tuples \{content,cache\} $\mathcal{K}\times\mathcal{M}$.

\noindent\underline{Item (2):} We consider the objective function of \eq{eq:objective-femto-related} and apply the same steps as in proof of Lemma~\ref{lemma:single-submodular}. Specifically, for all sets $A\subseteq B\subset \mathcal{K}\times\mathcal{M}$ and \{content, SC\} tuples $(\ell,m)\in V\backslash B$, we get
\begin{footnotesize}
\begin{align*}
 f&\left(A\cup \{(\ell,m)\} \right) - f\left(A \right) =\\
 =&	\sum_{i=1}^{N}\sum_{k=1}^{K} p_{k}^{i} \left(1 - \prod_{(n,j)\in A\cup \{(\ell,m)\}} \left(1- u_{kn}^{i}\cdot q_{ij}\right)\right)\\
&	- \sum_{i=1}^{N}\sum_{k=1}^{K} p_{k}^{i}  \left(1 - \prod_{(n,j)\in A} \left(1- u_{kn}^{i}\cdot q_{ij}\right)\right)\\
 =& \sum_{i=1}^{N}\sum_{k=1}^{K} p_{k}^{i}\cdot  u_{k\ell}^{i}\cdot q_{im}\cdot \prod_{(n,j)\in A} \left(1- u_{kn}^{i}\cdot q_{ij}\right)\geq 0
\end{align*}
\end{footnotesize}
which proves \textit{submodularity}, and 
\begin{footnotesize}
\begin{align*}
&f(B)-f(A)  =\\
&\hspace{0.5cm}\sum_{i=1}^{N}\sum_{k=1}^{K} p_{k}^{i}\left(1 - \prod_{(n,j)\in B} \left(1- u_{kn}^{i} q_{ij}\right)\right) -\\ &\hspace{1.5cm}\sum_{i=1}^{N}\sum_{k=1}^{K} p_{k}^{i}\left(1 - \prod_{(n,j)\in A} \left(1- u_{kn}^{i} q_{ij}\right)\right)=\\
&\sum_{i=1}^{N}\sum_{k=1}^{K} p_{k}^{i} \prod_{(n,j)\in A} \left(1- u_{kn}^{i}\cdot q_{ij}\right) \cdot \left(1-\prod_{(n,j)\in B\backslash A} \left(1- u_{kn}^{i}\cdot q_{ij}\right)\right)\geq0 
\end{align*}
\end{footnotesize}
which proves \textit{monotonicity}.

%We will consider the objective function of \eq{eq:objective-femto-related} and compare it that of \eq{eq:objective-single-cache}. The term in the parenthesis now contains a double product (over both neighbor SCs and over related contents). However, this product can be expanded to single product of terms $(1-{u_{kn}^{i}} \cdot x_{nj} \cdot q_{ij})$. Hence, we can apply again the same steps as in the proof of Lemma~\ref{lemma:single-submodular} to show that the inner sum (over content index $k$) is submodular and monotone. The outer sum in the objective is just a sum of such functions, and the sum of submodular functions is known to be submodular.

To show that the constraint is a matroid (see e.g.~\cite{krause2012submodular} for the definition of a matroid), we consider the set $\mathcal{V} = \mathcal{K}\times\mathcal{M}$ (i.e., all the possible tuples \{content,cache\}) and the collection of subsets of $\mathcal{V}$ that do not violate the capacity of the caches 
\begin{equation*}
{I} = \left\{S\subseteq 2^{\mathcal{V}}: |S\cap 2^{\{\mathcal{K},m\}}|\leq C,~~\forall m\in\mathcal{M}\right\}
\end{equation*}
Then:\\
\noindent (a) For all sets $A$ and $B$ that $A\subseteq B \subseteq \mathcal{V}$, it holds that if $B\subseteq \mathcal{I}$ (i.e., the caching placement defined by $B$ does not violate the size of the caches) then $A\subseteq \mathcal{I}$, because in $A$ every cache has to store the same or less content than in $B$ and thus no capacity constraint is violated.\\
\noindent (b) For all sets $A,B\in\mathcal{I}$ (i.e., feasible caching placements) and $|B|>|A|$ (i.e., in $B$ more contents are cached), $\exists \ell\in B\backslash A$ that $A\cup\{\ell\}\in\mathcal{I}$, since in $A$ not all caches are full (otherwise $B$ would violate the capacity constraint, i.e., $B\notin\mathcal{I}$), which means that there exists at least one more content can be cached (and this content can be selected to be from the set $B$).

From (a) and (b), it follows directly that the constraint is a matroid~\cite{krause2012submodular}.

\section{Proof of Corollary~\lowercase{\ref{thm:corollary-femto-second-model-subcases}}}\label{app:corollary-subcases-second-model}

\textbf{Sub-case 1.} Similar to the proof Corollary~\ref{thm:corollary-single-cache-subcases}, by taking the conditional expectations on the different $F_{kn}(x)$ in \eq{eq:objective-femto-second-model}, we get:
\begin{align}
SCHR = {\sum_{i=1}^{N}}\sum_{k=1}^{K} p_{k}^{i} \cdot E\left[ \max_{n\in\mathcal{K}, j\in\mathcal{M}} \left(u_{kn}^{i} \cdot x_{nj}\cdot q_{ij}\right)\right] \label{eq:expectation-max}
\end{align}
Let us consider the random variable $Y_{kS}$, where $Y_{kS}= \max_{(n,j)\in S} \left(u_{kn}^{i}\right)$, where $S=\{(\ell,m):\ell\in\mathcal{K}, m\in\mathcal{M},x_{\ell m}=1\}$. The distribution of $Y_{kS}$ (as the max value of independent random variables) is given by
\begin{equation}
F_{kS}(x) = P\{Y_{kS}\leq x\}= \prod_{(n,j)\in S} F_{kn}(x)
\end{equation}
Then \eq{eq:expectation-max} becomes
\begin{align*}
SCHR 
&= {\sum_{i=1}^{N}}\sum_{k=1}^{K} p_{k}^{i} \cdot E\left[ Y_{kS}\cdot q_{ij}\right] \\
&= {\sum_{i=1}^{N}}\sum_{k=1}^{K} p_{k}^{i} \cdot E\left[ Y_{kS}\right]\cdot q_{ij} \\
&= {\sum_{i=1}^{N}}\sum_{k=1}^{K} p_{k}^{i} \cdot \left(\int \left(1-F_{kS}(t)\right)dt\right)\cdot q_{ij} \\
&= {\sum_{i=1}^{N}}\sum_{k=1}^{K} p_{k}^{i} \cdot \left(\int \left(1-\prod_{(n,j)\in S} F_{kn}(t)\right)dt\right)\cdot q_{ij}
\end{align*}

\textbf{Sub-case 2} follows straightforwardly.

\section{Proof of Lemma~\lowercase{\ref{lemma:np-and-submodular-second-model}}}\label{app:lemma-second-model}
\underline{Item (1):} Optimization Problem~\ref{problem:femto-second-utility-model} is of the exact same nature as Optimization Problem~\ref{problem:femto-related}, so it follows that it is NP-hard.
%%%%

\noindent\underline{Item (2):} We proceed similarly to the proof of Lemma~\ref{lemma:single-submodular}. The objective function of \eq{eq:objective-femto-second-model} $f(X):\{0,1\}^{K\times M}\rightarrow \mathbb{R}$ is equivalent to a set function $f(S):2^{\mathcal{K}\times \mathcal{M}}\rightarrow \mathbb{R}$, where $\mathcal{K}$ and $\mathcal{M}$ are the finite \emph{ground sets} of contents and SCs, respectively, and $S = \{k \in \mathcal{K}, j\in\mathcal{M}: x_{kj}=1\}$:
\begin{equation}
f(S)\equiv\sum_{i=1}^{N} \sum_{k=1}^{K} p_{k}^{i}\cdot \max_{(n,j)\in S} \left(u_{kn}^{i} \cdot q_{ij}\right)   
\end{equation}

For all sets $A\subseteq B\subset V$ and \{content, SC\} tuples $(\ell,m)\in V\backslash B$, we get
%\begin{small}
\begin{align*}
& f\left(A\cup \{(\ell,m)\} \right) - f\left(A \right) =\\
& =	\sum_{i=1}^{N}\sum_{k=1}^{K} p_{k}^{i} \max_{(n,j)\in A\cup \{(\ell,m)\}} \left(u_{kn}^{i} q_{ij}  \right)  \\
&\hspace{1.5cm}
	- \sum_{i=1}^{N}\sum_{k=1}^{K} p_{k}^{i} q_{i}  \max_{(n,j)\in A} \left(u_{kn}^{i} q_{ij}  \right)
\\
& = \sum_{i=1}^{N}\sum_{k=1}^{K} p_{k}^{i} R\left(u_{k\ell}^{i}\cdot q_{im}-\max_{(n,j)\in A} \left(u_{kn}^{i}q_{ij}\right)\right)
\end{align*}
%\end{small}
where in the last equation we use the \textit{ramp function} defined as $R(x) = x$ for $x\geq0$ and $R(x)=0$ for $x<0$. 
%\begin{equation*}
%\textstyle
%R(x) = \left\{ 
%\begin{tabular}{lc}
%$x$  &  $x\geq0$\\
%$0$  &  $x<0$
%\end{tabular}
%\right.
%\end{equation*}
Subsequently,
%\begin{small}
\begin{align*}
\left[f\left(A\cup \{(\ell,m)\} \right) - f\left(A \right)\right] &- \left[f\left(B\cup \{(\ell,m)\} \right) - f\left(B \right)\right] =\\
=\sum_{i=1}^{N}\sum_{k=1}^{K} &p_{k}^{i} \left[R\left(u_{k\ell}^{i} q_{im}-\max_{(n,j)\in A} \left(u_{kn}^{i}q_{ij}\right)\right)\right.\\
&\left.	-  R\left(u_{k\ell}^{i}q_{im}-\max_{(n,j)\in B} \left(u_{kn}^{i}q_{ij}\right)\right)\right]
\end{align*}
%\end{small}
The above equation is always $\geq0$ (which proves that the objective function \eq{eq:objective-femto-second-model} is \textit{submodular}), since the ramp function is monotonically increasing and comparing the two arguments of the function $R(x)$ in the above equation, gives
%\begin{small}
\begin{align*}
u_{k\ell}^{i}q_{im}-\max_{(n,j)\in A} \left(u_{kn}^{i}q_{ij}\right) - \left(u_{k\ell}^{i}q_{im}-\max_{(n,j)\in B} \left(u_{kn}^{i}q_{ij}\right)\right) =\\
\max_{(n,j)\in B} \left(u_{kn}^{i}q_{ij}\right) - \max_{(n,j)\in A} \left(u_{kn}^{i}q_{ij}\right)\geq 0
\end{align*}
%\end{small}
since $B$ is a superset of $A$ and therefore its maximum will be at least equal or greater than the maximum value in set $A$.

Similarly, since $A\subseteq B$ it holds
\begin{footnotesize}
\begin{align*}
f(B)-f(A)  =\sum_{i=1}^{N}\sum_{k=1}^{K} p_{k}^{i} \left(\max_{(n,j)\in B} \left(u_{kn}^{i}q_{ij}\right) - \max_{(n,j)\in A} \left(u_{kn}^{i}q_{ij}\right)\right) \geq0
\end{align*}
\end{footnotesize}
which proves that the \eq{eq:objective-femto-second-model} is \textit{monotone}.

%%%% extra supplementary material %%%
%
%\clearpage \newpage
\newpage

\section{(1-1/\textit{e}) approximation algorithm for Optimization Problem~\lowercase{\ref{problem:single-knapsack}}}\label{app:algorithm-single-1/e}

\begin{algorithm}%[h]

\caption{$\left(1-\frac{1}{e}\right)$-approximation Algorithm for Optimization Problem~\ref{problem:single-knapsack}.}
\begin{algorithmic}[1]
\State $A\leftarrow \left\{S\subseteq \mathcal{K}:~|S|<3~\text{and}~\sum_{i\in S}s_{i}\leq C\right\}$
\State $S^{(1)}\leftarrow  \underset{S \in A}{argmax} f(S)$
\Statex
\State $B\leftarrow \left\{S\subseteq \mathcal{K}:~|S|=3~\text{and}~\sum_{i\in S}s_{i}\leq C\right\}$
\State $S^{(2)}\leftarrow \emptyset$
\For {$S\in B$}
	\State $U\leftarrow \mathcal{K}\backslash S$
	\State $W\leftarrow$ list$(w_{i}),~\forall i\in U$
	\State $H\leftarrow$\textproc{ModifiedGreedy}($U$,[W])
	\If {$f(H)>f(S^{(2)})$}
		\State $S^{(2)}\leftarrow H$
	\EndIf 
\EndFor
\Statex
\If{$f(S^{(1)})>f(S^{(2)})$}
	\State $S^{*}\leftarrow  S^{(1)}$
\Else
	\State $S^{*}\leftarrow  S^{(2)}$
\EndIf
\State\Return $S^{*}$
\end{algorithmic}
\label{alg:complex-single-knapsack}
\end{algorithm}

%\add{[Theo text]}
%\input{../reports/femto-theory_max_proof_TG}

\end{document}